\documentclass[a4paper,UKenglish,cleveref,hyperref,autoref]{lipics-v2021}

\title{Twin-width and polynomial kernels}
\titlerunning{Twin-width and polynomial kernels}

\author{\'{E}douard Bonnet}{Univ Lyon, CNRS, ENS de Lyon, Université Claude Bernard Lyon 1, LIP UMR5668, France \and \url{http://perso.ens-lyon.fr/edouard.bonnet/}}{edouard.bonnet@ens-lyon.fr}{https://orcid.org/0000-0002-1653-5822}{}
\author{Eun Jung Kim}{Universit\'{e} Paris-Dauphine, PSL University, CNRS UMR7243, LAMSADE, Paris, France}{eun-jung.kim@dauphine.fr}{https://orcid.org/0000-0002-6824-0516}{}
\author{Amadeus Reinald}{Univ Lyon, CNRS, ENS de Lyon, Université Claude Bernard Lyon 1, LIP UMR5668, France}{amadeus.reinald@ens-lyon.fr}{}{}
\author{St\'{e}phan Thomass\'{e}}{Univ Lyon, CNRS, ENS de Lyon, Universit\'{e} Claude Bernard Lyon 1, LIP UMR5668, France}{stephan.thomasse@ens-lyon.fr}{}{}
\author{R\'{e}mi Watrigant}{Univ Lyon, CNRS, ENS de Lyon, Universit\'{e} Claude Bernard Lyon 1, LIP UMR5668, France}{remi.watrigant@univ-lyon1.fr}{https://orcid.org/0000-0002-6243-5910}{}

\authorrunning{\'E. Bonnet, E. J. Kim, A. Reinald, S. Thomassé, R. Watrigant}

\Copyright{Édouard Bonnet, Eun Jung Kim, Amadeus Reinald, Stéphan Thomassé, Rémi Watrigant}



\category{}

\relatedversion{}

\supplement{}

\acknowledgements{}

\nolinenumbers 

\hideLIPIcs  

\EventEditors{John Q. Open and Joan R. Access}
\EventNoEds{2}
\EventLongTitle{42nd Conference on Very Important Topics (CVIT 2016)}
\EventShortTitle{CVIT 2016}
\EventAcronym{CVIT}
\EventYear{2016}
\EventDate{December 24--27, 2016}
\EventLocation{Little Whinging, United Kingdom}
\EventLogo{}
\SeriesVolume{42}
\ArticleNo{23}

\usepackage[utf8]{inputenc}  

\usepackage[T1]{fontenc}
\usepackage{lmodern}

\usepackage[colorinlistoftodos,bordercolor=orange,backgroundcolor=orange!20,linecolor=orange,textsize=normalsize]{todonotes}

\usepackage{amsmath}  
\usepackage{amssymb}     
\usepackage{bbm}
\usepackage{accents}
\usepackage{complexity}

\usepackage{booktabs}
\usepackage{hyperref}
\usepackage{mathrsfs}
\usepackage{paralist}
\usepackage{fixmath}

\crefformat{equation}{\textup{#2(#1)#3}}
\crefrangeformat{equation}{\textup{#3(#1)#4--#5(#2)#6}}
\crefmultiformat{equation}{\textup{#2(#1)#3}}{ and \textup{#2(#1)#3}}
{, \textup{#2(#1)#3}}{, and \textup{#2(#1)#3}}
\crefrangemultiformat{equation}{\textup{#3(#1)#4--#5(#2)#6}}%
{ and \textup{#3(#1)#4--#5(#2)#6}}{, \textup{#3(#1)#4--#5(#2)#6}}{, and \textup{#3(#1)#4--#5(#2)#6}}

\Crefformat{equation}{#2Equation~\textup{(#1)}#3}
\Crefrangeformat{equation}{Equations~\textup{#3(#1)#4--#5(#2)#6}}
\Crefmultiformat{equation}{Equations~\textup{#2(#1)#3}}{ and \textup{#2(#1)#3}}
{, \textup{#2(#1)#3}}{, and \textup{#2(#1)#3}}
\Crefrangemultiformat{equation}{Equations~\textup{#3(#1)#4--#5(#2)#6}}%
{ and \textup{#3(#1)#4--#5(#2)#6}}{, \textup{#3(#1)#4--#5(#2)#6}}{, and \textup{#3(#1)#4--#5(#2)#6}}

\makeatletter
\newtheorem*{rep@theorem}{\rep@title}
\newcommand{\newreptheorem}[2]{%
\newenvironment{rep#1}[1]{%
 \def\rep@title{#2 \ref{##1}}%
 \begin{rep@theorem}}%
 {\end{rep@theorem}}}
\makeatother

\newreptheorem{theorem}{Theorem}
\newreptheorem{lemma}{Lemma}

\usepackage{xspace}
\usepackage{tikz}

\usepackage[ruled,vlined,linesnumbered]{algorithm2e}

\usetikzlibrary{fit}
\usetikzlibrary{arrows}
\usetikzlibrary{patterns}
\usetikzlibrary{calc}
\usetikzlibrary{shapes}
\usetikzlibrary{positioning}
\usetikzlibrary{math}
\usetikzlibrary{shapes.symbols}
\usetikzlibrary{decorations.pathreplacing}
\tikzset{draw half paths/.style 2 args={%
  decoration={show path construction,
    lineto code={
      \draw [#1] (\tikzinputsegmentfirst) -- 
         ($(\tikzinputsegmentfirst)!0.5!(\tikzinputsegmentlast)$);
      \draw [#2] ($(\tikzinputsegmentfirst)!0.5!(\tikzinputsegmentlast)$)
        -- (\tikzinputsegmentlast);
    }
  }, decorate
}}

\usepackage[scr=boondox,scrscaled=1.05]{mathalfa}

\renewcommand{\geq}{\geqslant}
\renewcommand{\leq}{\leqslant}

\newcommand{\kds}{\textsc{$k$-Dominating Set}\xspace}
\newcommand{\ds}{\textsc{Dominating Set}\xspace}
\newcommand{\kis}{\textsc{$k$-Independent Set}\xspace}

\newcommand{\kconvc}{\textsc{Connected $k$-Vertex Cover}\xspace}
\newcommand{\kcapvc}{\textsc{Capacitated $k$-Vertex Cover}\xspace}

\newcommand{\lk}{LK\xspace}
\newcommand{\pk}{PK\xspace}

\newcommand{\defparproblem}[4]{
 \vspace{1mm}
\noindent\fbox{
 \begin{minipage}{0.96\textwidth}
 \begin{tabular*}{\textwidth}{@{\extracolsep{\fill}}lr} #1 & {\bf{Parameter:}} #3 \\ \end{tabular*}
 {\bf{Input:}} #2 \\
 {\bf{Question:}} #4
 \end{minipage}
 }
 \vspace{1mm}
}

\theoremstyle{definition}

\newtheorem{reduction}{Reduction Rule}

\newenvironment{proofofclaim}{\noindent \textsc{Proof of the Claim:}}{\hfill$\Diamond$\medskip}

\newcommand{\tww}{tww}

\newcommand{\cR}{\mathcal{R}}
\newcommand{\cQ}{\mathcal{Q}}
\newcommand{\cB}{\mathcal{B}}

\renewcommand{\cL}{\mathcal{L}} 
\colorlet{npink}{red!30!pink}

\newcommand\abs[1]{\lvert #1\rvert}

\newenvironment{edouard}[1]{\textcolor{blue}{Édouard:~#1}}{}
\newenvironment{ejk}[1]{\textcolor{orange}{EJK:~#1}}{}

\begin{document}

\maketitle

\begin{abstract}
  We study the existence of polynomial kernels, for parameterized problems without a polynomial kernel on general graphs, when restricted to graphs of bounded twin-width.
  Our main result is that a polynomial kernel for \textsc{$k$-Dominating Set} on graphs of twin-width at most 4 would contradict a standard complexity-theoretic assumption.
  The reduction is quite involved, especially to get the twin-width upper bound down to 4, and can be tweaked to work for \textsc{Connected $k$-Dominating Set} and \textsc{Total $k$-Dominating Set} (albeit with a worse upper bound on the twin-width).
  The \textsc{$k$-Independent Set} problem admits the same lower bound by a much simpler argument, previously observed [ICALP '21], which extends to \textsc{$k$-Independent Dominating Set}, \textsc{$k$-Path}, \textsc{$k$-Induced Path}, \textsc{$k$-Induced Matching}, etc.
  
  On the positive side, we obtain a simple quadratic vertex kernel for \textsc{Connected $k$-Vertex Cover} and \textsc{Capacitated $k$-Vertex Cover} on graphs of bounded twin-width.
  Interestingly the kernel applies to graphs of Vapnik-Chervonenkis density~1, and does not require a witness sequence.
  We also present a more intricate $O(k^{1.5})$ vertex kernel for~\textsc{Connected $k$-Vertex Cover}.
  Finally we show that deciding if a graph has twin-width at most~1 can be done in polynomial time, and observe that most optimization/decision graph problems can be solved in polynomial time on graphs of twin-width at most~1.
\end{abstract}
\maketitle

\section{Introduction}\label{sec:intro}

The \emph{twin-width} of a graph can be defined in the following way.
A~\emph{partition sequence} of an $n$-vertex graph $G$, is a sequence $\mathcal P_n, \ldots, \mathcal P_1$ of partitions of its vertex set $V(G)$, such that $\mathcal P_n$ is the set of singletons $\{\{v\}~:~v \in V(G)\}$, $\mathcal P_1$ is the singleton set $\{V(G)\}$, and for every $2 \leqslant i \leqslant n$, $\mathcal P_{i-1}$ is obtained from $\mathcal P_i$ by merging two of its parts into one.
Two parts $P, P'$ of a same partition $\mathcal P$ of $V(G)$ are said \emph{homogeneous} if either every pair of vertices $u \in P, v \in P'$ are non-adjacent, or every pair of vertices $u \in P, v \in P'$ are adjacent.
Finally the twin-width of $G$ is the least integer $d$ such that there is partition sequence $\mathcal P_n, \ldots, \mathcal P_1$ of $G$ with every part of every $\mathcal P_i$ ($1 \leqslant i \leqslant n$) being homogeneous to every other parts of $\mathcal P_i$ but at most~$d$.
We call such a partition sequence a~\emph{$d$-sequence}.

On the one hand, a surprisingly wide variety of graphs have low twin-width.
Graph classes with bounded twin-width include classes with bounded treewidth, or even rank-width, proper minor-closed classes, every hereditary proper subclass of permutation graphs, bounded-degree string graphs~\cite{twin-width1}, classes with bounded queue or stack number, some expander families~\cite{twin-width2}.
Furthermore on those particular classes, we can find (non necessarily optimum) $O(1)$-sequences in polynomial time.
We observe that such an approximation algorithm is still missing in general graphs, but exists for \emph{ordered} binary structures~\cite{twin-width4}.

On the other hand, bounded twin-width classes have interesting algorithmic and structural properties.
Remarkably, given a partition sequence witnessing that an $n$-vertex graph $G$ has twin-width at most~$d$, and a first-order sentence $\varphi$, one can decide if $\varphi$ holds in $G$ in time $f(|\varphi|,d)\,n$~for a computable, but non-elementary, function~$f$~\cite{twin-width1}.
That general framework is called \emph{first-order model checking}, and generalizes problems like \kis with $\varphi = \exists x_1 \ldots \exists x_k \bigwedge_{1 \leqslant i < j \leqslant k} \neg (x_i = x_j \vee E(x_i,x_j))$ and \kds with $\varphi = \exists x_1 \ldots \exists x_k \forall x \bigvee_{1 \leqslant i \leqslant k} (x = x_i \vee E(x,x_i))$.
For these two particular problems, though, a~much better running time of $2^{O_d(k)}\,n$ is possible~\cite{twin-width3}.
In contrast, an algorithm running in time $f(k)n^{o(k)}$ for either of these problems on \emph{general} graphs, with $f$ being any computable function, would imply the improbable (or at least breakthrough) result that \textsc{$3$-SAT} can be solved in subexponential time~\cite{Chen06}.

Now we know that \kis and \kds are fixed-parameter tractable (\FPT), i.e., solvable in time $f(k)\,n^{O(1)}$, on graphs of bounded twin-width given with an $O(1)$-sequence, one can then ask whether polynomial kernels exist.
A \emph{kernel} is a polytime algorithm that produces, given an instance of a parameterized problem $\Pi$, an equivalent instance of $\Pi$ (i.e., the output is a YES-instance if and only if the input is a YES-instance) of size only function of the parameter.
A~\emph{polynomial kernel} is a kernel for which the latter function is polynomial.
Any decidable problem with a kernel is \FPT, and any {\FPT} problem admits a kernel.
However not every {\FPT} problem is believed to have a polynomial kernel.
And indeed such an outcome would imply an unlikely collapse of complexity classes.

We already observed that there is a constant $d$ such that \kis is highly unlikely to have a polynomial kernel on graphs with twin-width at most~$d$~\cite{twin-width3}.
The OR-composition\footnote{See~\cref{sec:prelim} for the relevant background on how to rule out a polynomial kernel} is straightforward from the following facts: (1) cliques have twin-width~0 and planar graphs have bounded twin-width~\cite{twin-width1}, (2) the twin-width of every graph is the maximum twin-width of its modules and quotient graph (see~\cref{lem:modular}), and (3) \textsc{Maximum Independent Set} is NP-hard in (subcubic) planar graphs~\cite{Garey76}.
Then one can blow every vertex of a clique $K_t$ into a distinct graph among $t$ planar \textsc{Maximum Independent Set}-instances.
Facts (1) and (2) imply that the constructed graph has bounded twin-width, while the correctness of the OR-composition is easy to check.
Incidentally the exact same reduction rules a polynomial kernel out for~\textsc{$k$-Independent Dominating Set}.
Furthermore \textsc{Minimum Independent Dominating Set} is NP-hard in grid graphs~\cite{Clark90}, and \textsc{Maximum Independent Set} is NP-hard in subdivisions of grid graphs (since these coincide with planar graphs of degree at most~4).
Since these graphs have twin-width at most 4 (see~\cref{lem:subdivided-grid}), no polynomial kernel is likely to exist for both problems (even when a 4-sequence is given in the input).
It should be noted that this simple reduction fails for \kds: one can dominate the constructed graph by picking only two vertices (from two distinct instances).

The parameterized complexity ({\FPT} algorithms and kernels) of \kds\footnote{All the subsequent results also hold for \kis.} on ``sparse''\footnote{\emph{Sparse} is an overloaded term; here we use it as \emph{not containing arbitrarily large bicliques as subgraphs}.} classes has a rich and interesting history.
Subexponential {\FPT} algorithms with running time $2^{O(\sqrt k)}\,n^{O(1)}$ are known in planar graphs~\cite{Fomin06,Dorn06}, bounded-genus graphs and more generally classes excluding a fixed minor~\cite{Demaine05,Fomin18,Gutner09}, and an {\FPT} algorithm with running time $2^{O(k)}\,n$ exists in classes excluding a fixed topological minor~\cite{AlonG09}.
On these classes the mere existence of an {\FPT} algorithm (but not the particular, enhanced running time) is subsumed by an algorithmic meta-theorem of Grohe, Kreutzer, and Siebertz~\cite{Grohe17} that says that first-order model checking is {\FPT} in any nowhere dense class.\footnote{The definition of nowhere denseness being technical and unnecessary to the current paper, we refer the interested reader to~\cite{sparsity}.
Let us just mention that bounded-degree graphs, planar graphs, and proper (topological) minor-closed classes are all nowhere dense.}
More general than nowhere dense classes are bounded-degeneracy graphs, or further, $K_{t,t}$-free classes, i.e., excluding the biclique $K_{t,t}$ as a subgraph.
Alon and Gutner~\cite{AlonG09} give an {\FPT} algorithm in $d$-degenerate graphs running in time $k^{O(dk)}\,n$.
And Philip, Raman, and Sikdar~\cite{Philip12} extend the fixed-parameter tractability of \kds to any $K_{t,t}$-free class (for a fixed $t$).
Telle and Villanger~\cite{Telle19} further show that \kds on $K_{t,t}$-free graphs is {\FPT} for the combined parameter $k+t$.

In parallel to these algorithms, the existence of polynomial, or even linear, kernels have been thoroughly investigated.
In 2004, Alber, Fellows, and Niedermeier~\cite{Alber04} presented a linear kernel for \kds on planar graphs that triggered a series of works.
Linear kernels are known on planar graphs~\cite{Alber04,Chen07}, bounded-genus graphs~\cite{Fomin04}, apex-minor-free graphs~\cite{bidim-kernel}, but more generally in any class excluding a fixed topological minor~\cite{Fomin18}.
\kds admits a polynomial kernel on graphs of girth 5 (that is, excluding the triangle and the biclique $K_{2,2}$ as a subgraph)~\cite{Raman08}.
A polynomial kernel of size $O(k^{(t+1)^2})$ is obtained for $K_{t,t}$-free graphs~\cite{Philip12}, the most general ``sparse'' class.
Contrary to the {\FPT} algorithm, a polynomial kernel in the parameter $k+t$ is highly unlikely~\cite{Dom14}.
More precisely, for any $\varepsilon > 0$, a kernel of size $k^{(t-1)(t-3)-\varepsilon}$ would imply that {\coNP} $\subseteq$ \NP/poly~\cite{CyganGH17}.
On classes of bounded expansion\footnote{We will not need a definition of \emph{expansion} here. Bounded expansion classes are more general than topological-minor-free classes and less general than nowhere dense classes.} \kds has a linear kernel, while the seemingly closely related \textsc{Connected $k$-Dominating Set} has no polynomial kernel~\cite{Drange16}.
The latter result refines a reduction showing the same lower bound on 2-degenerate graphs~\cite{CyganPPW12}.

Beyond sparse classes, for which most answers turn out positive, the parameterized complexity of \kds seems to conceal many surprises, some of which recently unraveled. 
We already mentioned that \kds is {\FPT} on bounded twin-width graphs given with an $O(1)$-sequence.
Let us also mention that the same problem is actually \W$[1]$-hard (hence unlikely {\FPT}) on circle graphs~\cite{Bousquet14}.
This is somewhat unexpected since \textsc{Dominating Set} is polytime solvable on permutation graphs~\cite{Farber85}, a large subclass of circle graphs.
On the positive side, \kds admits a polynomial kernel on so-called \emph{$c$-closed graphs}~\cite{Koana20}, a far-reaching dense generalization of bounded $d$-degenerate graphs. 

\paragraph*{Our results.}
We are back to wondering whether \kds admits a polynomial kernel on graphs given with an $O(1)$-sequence.
On the one hand, a polynomial kernel would ``fit all the data points'' considering that the examples of bounded twin-width classes previously given 
are either $K_{t,t}$-free (and one concludes with~\cite{Philip12}) or are dense classes on which \textsc{Minimum Dominating Set} is polytime solvable, like bounded rank-width graphs~\cite{Courcelle00}, and (subclasses of) permutation graphs~\cite{Farber85}.
On the other hand, the same could be said of \kis for which we already ruled out such a kernel.
Yet we will see in \cref{sec:outline} that the above OR-composition not working for \kds is part of a more general obstacle toward establishing its incompressibility.   
In the same section we lay down our plan to overcome that obstacle and show the following.

\begin{theorem}\label{thm:main}
Unless {\coNP} $\subseteq$ \NP/poly, \kds on graphs of twin-width at most 4 does not admit a polynomial kernel, even if a 4-sequence of the graph is given.
\end{theorem}

We mentioned that the same statement holds much more directly for \kis and \textsc{$k$-Independent Dominating Set}.
With analogous arguments, we can add \textsc{$k$-Path}, \textsc{$k$-Induced Path}, \textsc{$k$-Induced Matching} to the list. 
Local gadget modifications of the proof of~\cref{thm:main} yield the same kernel lower bound for variants of \kds such as \textsc{Connected $k$-Dominating Set} and \textsc{Total $k$-Dominating Set}, on graphs of bounded twin-width.
More work would be necessary to get the lower bound for twin-width at most~4. 

On the positive side, \kconvc and \kcapvc admit polynomial kernels on graphs of bounded twin-width, while such kernels are unlikely on general graphs~\cite{Dom14}.
Interestingly, our kernelization algorithm does not require an $O(1)$-sequence.
\begin{theorem}\label{thm:cnvc-cpvc}
  \kconvc and \kcapvc admit a kernel with $O(k^2)$ vertices on any class of bounded twin-width.
\end{theorem}
A linear kernel (in the number of vertices) is known for apex-minor-free classes~\cite{bidim-kernel} via the generic framework of bidimensionality, and even for topological-minor-free classes~\cite{Kim16}.
Another powerful meta-theorem by Gajarsk\'y et al.~\cite{Gajarsky17} says that every problem with the so-called~\emph{finite integer index} (intuitively, that its boundaried graphs provide finitely many distinct contexts) has a linear kernel on bounded expansion classes when parameterized by the vertex cover number (and more generally by the size of a smallest vertex subset whose deletion leaves the graph with bounded treedepth).
In particular this yields a linear kernel for \kconvc, further extending the two previous results.
Besides \kconvc has a polynomial kernel on $K_{t,t}$-free graphs~\cite{CyganPPW12}.


\cref{thm:cnvc-cpvc} is based on the following useful lemma stating that, in graphs of bounded twin-width, the number of distinct neighborhood traces inside a subset of vertices is at most linear in the size of the subset. 
\begin{lemma}\label{lem:vc-density-1}
  There is a function $f$ such that for every graph $G$ of twin-width~$d$ and $X \subseteq V(G)$, the number of distinct neighborhoods in $X$, $|\{N(v) \cap X~:~v \in V(G)\}|$, is at most $f(d)|X|$.
\end{lemma}
A more compact rewording, using the language of Vapnik-Chervonenkis parameters, is that the neighborhood set-system of graphs of bounded twin-width has VC density~1.
By extension, we will say that a graph class has VC density at most~1, if its neighborhood hypergraphs do.
That bounded twin-width classes have VC density~1 is an interesting property, that is shared with classes of bounded expansion.
For example it implies a constant-factor approximation for \textsc{Min Dominating Set} (obtained in a rather different manner in \cite{twin-width3}) via small $\varepsilon$-nets~\cite{Chan12}.
\cref{lem:vc-density-1} was independently obtained by Wojciech Przybyszewski in his master thesis~\cite{PrzTor}.

For \kconvc, an improved kernel can be obtained with a more elaborate argument.
\begin{theorem}\label{thm:cnvc-improved}
  \kconvc admits a kernel with $O(k^{1.5})$ vertices on classes with VC density at most~1.
\end{theorem}

\begin{table}[h!]
\begin{tabular}{clll}
  \toprule
                             & \kds                                       & \textsc{Connected $k$-DS}                         & \textsc{Connected $k$-VC} \\
  \midrule
  general                    & \W$[2]$-complete~\cite{df13}               & \W$[2]$-complete~\cite{df13}                      & \FPT~\cite{Cygan12}, no \pk~\cite{Dom14}\\
  bounded expansion          & \lk~\cite{Drange16}                        & \FPT~\cite{Dawar09}, no \pk~\cite{Drange16}       & \lk~\cite{Gajarsky17}                   \\
  bounded biclique           & \pk~\cite{Philip12}                        & \FPT~\cite{Telle19}, no \pk~\cite{CyganPPW12}     & \pk, no \lk~\cite{CyganPPW12}           \\
  bounded degeneracy         & \pk~\cite{Philip12}                        & \FPT~\cite{Golovach08}, no \pk~\cite{CyganPPW12}  & \pk, no \lk~\cite{CyganPPW12,CyganGH17} \\
  $K_{1,3}$-free              & \pk~\cite{Hermelin19}                      & \FPT, no \pk~\cite{Hermelin19}                    & \lk (trivial)                           \\
  $K_{1,4}$-free              & \W$[2]$-complete~\cite{Cygan11}            & \W$[2]$-complete~\cite{Cygan11}                   & \lk (trivial)                           \\
  \midrule
  bounded twin-width         & \FPT~\cite{twin-width1}, no \pk            & \FPT~\cite{twin-width1},  \textbf{no \pk}         & $O(k^{1.5})$-vertex kernel  \\
  twin-width at most 4       & \FPT~\cite{twin-width1}, \textbf{no \pk}   & \FPT~\cite{twin-width1}                           & $O(k^{1.5})$-vertex kernel  \\
  twin-width at most 1       & \textbf{in P}                              & \textbf{in P}                                     & \textbf{in P}              \\
  \midrule
  VC density at most 1       & no \pk                                     & no \pk                                            & $\mathbf{O(k^{1.5})}$\textbf{-vertex kernel}  \\
  \bottomrule
\end{tabular}
\caption{Kernelization results for arguably the three main problems without a polynomial kernel in general graphs, but an interesting story in sparse classes.
  \pk stands for \emph{polynomial kernel}, \lk for \emph{linear kernel} (in the number of vertices).
  The indicated lack of a kernel is under the assumption that {\coNP} $\subseteq$ \NP/poly.
  Our new results are in bold (the results without a reference nor in bold are consequences of results in bold).}
\label{tbl:results}
\end{table}

Finally we extend cograph recognizability (cographs are exactly the graphs with twin-width~0) and prove:
\begin{theorem}\label{thm:tww1}
  One can decide in polynomial time if a graph has twin-width at most~1.
\end{theorem}
In case the input graph has indeed twin-width at most~1, a 1-sequence is found in polynomial time.
Furthermore we observe that a wide class of graph problems is efficiently solvable on inputs of twin-width at most~1. 
See~\cref{tbl:results} for a summary of most of our results, together with the relevant pointers on other graph classes.

\section{Preliminaries}\label{sec:prelim}

We denote by $[i,j]$ the set of integers $\{i,i+1,\ldots, j-1, j\}$, and by $[i]$ the set of integers $[1,i]$.
If $\mathcal X$ is a set of sets, we denote by $\cup \mathcal X$ their union.
The notation $O_d(\cdot)$ gives an asymptotic behavior when $d$ is seen as a constant.

\subsection{Graph theory}\label{subsec:graph-theory}

Unless stated otherwise, all graphs are assumed undirected and simple, that is, they do not have parallel edges or self-loops.
We denote by $V(G)$ and $E(G)$, the set of vertices and edges, respectively, of a graph $G$. 
For $S \subseteq V(G)$, we denote the \emph{open neighborhood} (or simply \emph{neighborhood}) of $S$ by $N_G(S)$, i.e., the set of neighbors of $S$ deprived of $S$, and the \emph{closed neighborhood} of $S$ by $N_G[S]$, i.e., the set $N_G(S) \cup S$.
We simplify $N_G(\{v\})$ into $N_G(v)$, and $N_G[\{v\}]$ into $N_G[v]$.
We may omit the subscript when $G$ is clear from the context.

We denote by $G[S]$ the subgraph of $G$ induced by $S$, and $G - S := G[V(G) \setminus S]$.
An injective mapping $\eta: V(H) \to V(G)$ \emph{witnesses} that $H$ is a subgraph of $G$, if $uv \in E(H)$ implies $\eta(u)\eta(v) \in E(G)$.
A bijective mapping $\eta: V(H) \to V(G)$ \emph{witnesses} that $H$ is a \emph{spanning} subgraph of $G$, if $uv \in E(H)$ implies $\eta(u)\eta(v) \in E(G)$.

A \emph{connected subset} (or \emph{connected set}) $S \subseteq V(G)$ is one such that $G[S]$ is connected.
A~\emph{dominating set} is a set $S \subseteq V(G)$ such that $N[S]=V(G)$.
A~\emph{vertex cover} is a set $S \subseteq V(G)$ such that every edge of $G$ has at least one of its two endpoints in $S$.
An~\emph{independent set} is a set $S \subseteq V(G)$ such that $G[S]$ is edgeless.

For two disjoint sets $A, B \subseteq V(G)$, $E(A,B)$ denotes the set of edges in $E(G)$ with one endpoint in $A$ and the other one in $B$.
We also denote by $G[A,B]$ the bipartite graph $(A \cup B,E(A,B))$.
Two distinct vertices $u, v$ such that $N(u) = N(v)$ are called \emph{false twins}, and \emph{true twins} if $N[u] = N[v]$.
Two vertices are \emph{twins} if they are false twins or true twins.

A set $S \subseteq V(G)$ is a \emph{module} if every vertex outside $S$ is fully adjacent to $S$ or fully non-adjacent to $S$.
A module is said to be \emph{trivial} if it is a singleton or the entire vertex set, 
and a graph without non-trivial modules is called \emph{prime}. 
The vertex set of every graph $G$ can be partitioned into modules $H_1, H_2, \ldots, H_\ell$, some of which can be singletons.
$\mathcal H = \{H_1, H_2, \ldots, H_\ell\}$ is then called a~\emph{modular decomposition} of $G$. The~\emph{quotient graph} $G/\mathcal H$ of $G$ has vertex set $\mathcal H$ and an edge between $H_i$ and $H_j$ whenever every vertex of $H_i$ is adjacent to every vertex of $H_j$.
By definition of a module, $H_iH_j$ is not an edge of $G/\mathcal H$ if and only if there is no edge in $G$ between $H_i$ and $H_j$.
We say that $X \subseteq V(G)$ is a \emph{module relative to $Y \subseteq V(G) \setminus X$} if every vertex of $Y$ is either fully adjacent to $X$ or fully non-adjacent to $X$.
Hence $X$ is a module if it is a module relative to $V(G) \setminus X$.

The \emph{strict half-graph of height $t$} is (up to isomorphism) the graph with vertex set $\{a_1, \ldots, a_t, b_1, \ldots, b_t\}$ and edge set $\{a_ib_j~:~i < j, i \in [t], j \in [t]\}$.
One can see $\{a_1, \ldots, a_t\}$ \emph{oriented toward} $\{b_1, \ldots, b_t\}$ in their realization of the relation $<$ over the indices.
The \emph{$\ell$-cycle of strict half-graphs of height t} is (up to isomorphism) the graph with vertex set $\{a^p_1, \ldots, a^p_t~:~p \in [0,\ell-1]\}$ and edge set $\{a^p_ia^{p+1 \mod \ell}_j~:~i < j, i \in [t], j \in [t], p \in [0,\ell-1]\}$.
Informally it is the graph obtained from an $\ell$-vertex cycle by replacing every edge by a strict half-graph of height $t$ with a consistent, say, clock-wise orientation.
See~\cref{fig:shg} for an example of a 5-cycle of strict half-graphs of height 6.
A~\emph{strict half-graph} is, for some natural $t$, the strict half-graph of height $t$.
A~\emph{cycle of strict half-graphs} is, for some natural $\ell$, the $\ell$-cycle of strict half-graphs of same height.

\begin{figure}
  \centering
  \begin{tikzpicture}[scale=.92,vertex/.style={draw,circle,inner sep=0.08cm},picked/.style={fill,circle,opacity=0.5,inner sep=0.055cm}]
    \def\t{6}
    \def\k{5}
    \def\hs{1}
    \def\vs{0.5}
    \def\z{0.4}
    \pgfmathtruncatemacro\tm{\t-1}
    \pgfmathtruncatemacro\km{\k-1}
    \pgfmathtruncatemacro\kp{\k+1}
    \foreach \i in {1,...,\t}{
      \begin{scope}[yshift=\i * \vs cm]
        \foreach \j in {1,...,\k}{
          \begin{scope}[xshift=\j * \hs cm]
              \node[vertex] (b\i\j) at (0,0) {} ;
          \end{scope}
        }
      \end{scope}
    }
    \foreach \zz in {0,\kp}{
     \foreach \i in {1,...,\t}{
      \begin{scope}[yshift=\i * \vs cm, xshift=\zz * \hs cm]
              \node[circle,inner sep=0.08cm] (b\i\zz) at (0,0) {} ;
          \end{scope}
     }
    }

    \foreach \j [count=\jp from 2] in {1,...,\km}{
      \foreach \i in {1,...,\tm}{
        \pgfmathtruncatemacro\ip{\i+1}
        \foreach \iq in {\ip,...,\t}{
          \draw[thick] (b\i\j) -- (b\iq\jp) ;
        }
      }
    }
    \foreach \i in {1,...,\tm}{
        \pgfmathtruncatemacro\ip{\i+1}
        \foreach \iq in {\ip,...,\t}{
          \path[draw half paths={thick}{draw=none}] (b\i\k) -- (b\iq\kp) ;
          \path[draw half paths={draw=none}{thick}] (b\i0) -- (b\iq1) ;
        }
    }
    \draw[thin,dashed] (0.5 * \hs,0.7 * \vs) -- (0.5 * \hs,\t * \vs + 0.6 * \vs) ;
    \draw[thin,dashed] (0.5 * \hs + \k * \hs,0.7 * \vs) -- (0.5 * \hs + \k * \hs,\t * \vs + 0.6 * \vs) ; 
  \end{tikzpicture}
  \caption{A 5-cycle of strict half-graphs of height 6.}
  \label{fig:shg}
\end{figure}
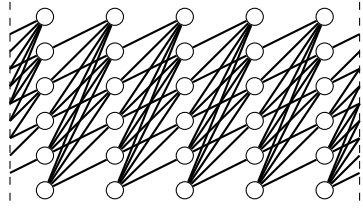

The \emph{$n \times m$~grid} is the graph with vertex set $[n] \times [m]$ and edges between any pair of vertices $(x,y), (x+1,y)$ or $(x,y), (x,y+1)$.
A~\emph{grid} is an $n \times m$~grid for some integer $n$ and $m$.
A~\emph{grid graph} is an induced subgraph of a grid.
To insist that we consider a grid and not a mere grid graph, we may use the term~\emph{complete grid}.

The \emph{neighborhood hypergraph} of a graph $G$ has vertex set $V(G)$ and edge set $\{N(v)~:~v \in V(G)\}$.
A family of hypergraphs $\mathcal H$ has \emph{Vapnik-Chervonenkis (VC) density at most~1} if there is a constant $c$ such that for every hypergraph $H \in \mathcal H$ and every $X \subseteq V(H)$, $|\{X \cap e~:~e \in E(H)\}| \leqslant c \cdot |X|$.
By extension, we may say that a graph class $\mathcal C$ has VC density at most~1 if the neighborhood hypergraph of every graph in $\mathcal C$ has VC density at most~1.  

\subsection{Contraction sequences and twin-width}\label{subsec:tww-def}

A \emph{trigraph $G$} has vertex set $V(G)$, black edge set $E(G)$, and red edge set $R(G)$, with $E(G)$ and $R(G)$ disjoint.
The~\emph{total graph} of trigraph $G$ is the graph $G'$ with $V(G') = V(G)$ and $E(G') = E(G) \cup R(G)$.
The~\emph{subtrigraph} of $G$ induced by $S$ is the trigraph $H$ with $V(H)=S$, $E(H)=E(G) \cap {S \choose 2}$, and~$R(H)=R(G) \cap {S \choose 2}$.
$H$~is then called an \emph{induced subtrigraph} of~$G$.

The \emph{set of neighbors $N_G(v)$} of a vertex $v$ in a trigraph $G$ consists of all the vertices adjacent to $v$ by a black or red edge.
A $d$-trigraph is a trigraph $G$ such that the \emph{red graph} $(V(G),R(G))$ has degree at most~$d$.
In that case, we also say that the trigraph has \emph{red degree} at most~$d$.
A \emph{contraction} or \emph{identification} in a trigraph~$G$ consists of merging two (non-necessarily adjacent) vertices $u$ and $v$ into a single vertex $z$, and updating the edges of $G$ in the following way.
Every vertex of the symmetric difference $N_G(u) \triangle N_G(v)$ is linked to $z$ by a red edge.
Every vertex $x$ of the intersection $N_G(u) \cap N_G(v)$ is linked to $z$ by a black edge if both $ux \in E(G)$ and $vx \in E(G)$, and by a red edge otherwise.
The rest of the edges (not incident to $u$ or $v$) remain unchanged.
See \cref{fig:contraction} for an illustration.
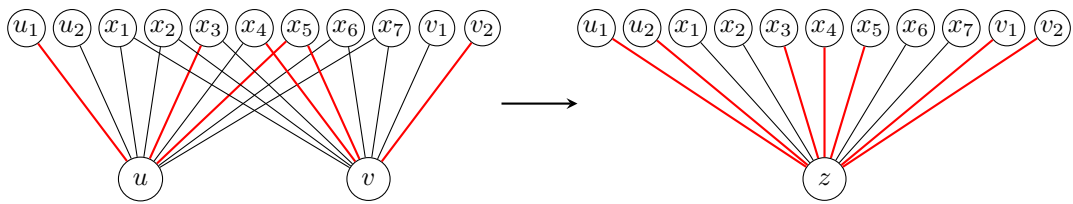
\begin{figure}[h!]
\begin{tikzpicture}
\def\v{2}
\def\t{6}
\def\s{0.6}

\draw[thick, -stealth] (3.25,\v /2) -- (4.25,\v/2) ;

\foreach \i/\j in {-5/u_1,-4/u_2,-3/x_1,-2/x_2,-1/x_3,0/x_4,1/x_5,2/x_6,3/x_7,4/v_1,5/v_2}{
  \node[draw,circle,inner sep=0.03cm] (n\i) at (\s * \i,\v) {$\j$} ; 
}

\node[draw,circle] (u) at (-1.5,0) {$u$} ;
\node[draw,circle] (v) at (1.5,0) {$v$} ;

\foreach \i in {-4,-3,-2,0,2,3}{
  \draw (u) -- (n\i) ;
}
\foreach \i in {-5,-1,1}{
  \draw[thick, red] (u) -- (n\i) ;
}
\foreach \i in {-3,...,-1,2,3,4}{
  \draw (v) -- (n\i) ;
}
\foreach \i in {0,1,5}{
  \draw[thick, red] (v) -- (n\i) ;
}

\begin{scope}[xshift=7.5cm]
\node[draw,circle] (uv) at (0,0) {$z$} ;
\foreach \i/\j in {-5/u_1,-4/u_2,-3/x_1,-2/x_2,-1/x_3,0/x_4,1/x_5,2/x_6,3/x_7,4/v_1,5/v_2}{
  \node[draw,circle,inner sep=0.03cm] (m\i) at (\s * \i,\v) {$\j$} ; 
}

\foreach \i in {-3,-2,2,3}{
  \draw (uv) -- (m\i) ;
}
\foreach \i in {-5,-4,-1,0,1,4,5}{
  \draw[thick, red] (uv) -- (m\i) ;
}
\end{scope}
\end{tikzpicture}
\caption{Contraction of vertices $u$ and $v$, and how the edges of the trigraph are updated.}
\label{fig:contraction}
\end{figure}

A \emph{$d$-sequence} (or \emph{contraction sequence}) is a sequence of \mbox{$d$-trigraphs} $G_n, G_{n-1}, \ldots, G_1$, where $G_n = G$, $G_1=K_1$ is the graph on a single vertex, and $G_{i-1}$ is obtained from $G_i$ by performing a single contraction of two (non-necessarily adjacent) vertices.
We observe that $G_i$ has precisely $i$ vertices, for every $i \in [n]$.
The twin-width of~$G$, denoted by $\tww(G)$, is the minimum integer~$d$ such that $G$ admits a~$d$-sequence.
Note that, in what precedes, the initial structure $G_n = G$ may be a trigraph instead of a graph.
Thus we defined twin-width more generally for trigraphs.
Similarly a \emph{partial $d$-sequence} from a $n$-vertex trigraph $G$ to an $i$-vertex trigraph $H$ is a sequence of \mbox{$d$-trigraphs} $G = G_n, G_{n-1}, \ldots, G_i = H$.
Observe that if $G$ has a partial $d$-sequence to $H$, and $H$ has itself a $d$-sequence, then the concatenation of these sequences is a $d$-sequence for $G$.

We may give a (partial) contraction sequence by listing the pairs of vertices to contract: $(u_1,v_1), (u_2,v_2), \ldots$ where $u_i, v_i$ are implicitly contracted to vertex $s(u_i) \cup s(v_i)$ with $s(x) = x$ if $x$ is a set of vertices and $s(x)=\{x\}$ if $x$ is a single vertex.
Thus for instance $u_2$ could be equal to $\{u_1,v_1\}$.

For $u \in V(G_i)$, we denote by $u(G)$ the subset of $V(G)$ that was contracted to the single vertex $u$ in $G_n, G_{n-1}, \ldots, G_i$.
Twin-width and $d$-sequences can be equivalently seen as a partition refinement process on $V(G)$.
We start with the finest partition $\mathcal P_n = \{\{v\}~:~v \in V(G)\}$, and end with the coarsest partition $\mathcal P_1 = \{V(G)\}$.
There is a \emph{partition sequence} $\mathcal P_n, \mathcal P_{n-1}, \ldots, \mathcal P_2, \mathcal P_1$ mimicking the contraction sequence, where the contraction of $u, v \in V(G_i)$ corresponds to the merge of parts $u(G_i), v(G_i) \in \mathcal P_i$ to form the part $u(G_i) \cup v(G_i) = z(G_{i-1}) \in \mathcal P_{i-1}$, while all the other parts are unchanged from $P_i$ to $P_{i-1}$.
The red degree (bounded by~$d$) of a part $P \in \mathcal P_i$ now corresponds to the number of other parts $P' \in \mathcal P_i$ which are not fully adjacent nor fully non-adjacent to $P$ in $G$.
We may denote by $G_{\mathcal P}$ the trigraph corresponding to partition $\mathcal P$ over $V(G)$.
Thus $G_i = G_{\mathcal P_i}$.

Given a partition $\mathcal H = \{H_1, H_2, \ldots, H_\ell\}$ of the vertex set $V(G)$ of a trigraph $G$, we call \emph{quotient trigraph $G/\mathcal H$} the trigraph obtained by contracting every $H_i$ into single vertices.
Note that the resulting trigraph does not depend on the order in which the contractions are made.
Thus the quotient trigraph is well-defined.

\subsection{Useful twin-width bounds}

It can be seen that the twin-width may only decrease when taking induced subtrigraphs. 
\begin{observation}\label{obs:induced-subgraph}
  Let $G$ be a trigraph and $H$ be an induced subtrigraph of $G$.
  Then, $\tww(H) \leqslant \tww(G)$.
\end{observation}

The following observation states that turning some non-edges or black edges into red edges can only increase the twin-width. 
\begin{observation}\label{obs:monotone}
  Let $G, G'$ be two trigraphs such that $V(G)=V(G')$, $R(G) \subseteq R(G')$, $E(G') \subseteq E(G)$, and $R(G) \cup E(G) \subseteq R(G') \cup E(G')$.
  Then $\tww(G) \leqslant \tww(G')$.
\end{observation}
\begin{proof}Indeed, any $d$-sequence for $G'$ is a $d$-sequence for $G$.\end{proof}

We will not need the next lemma in its particular form.
Yet it yields some good insight on contraction sequences in general, and on the ``core structure'' of low twin-width used in the proof~\cref{thm:main}, in particular.
Thus we give a short proof here.

\begin{lemma}\label{lem:strict-half-graphs}
  Cycles of strict half-graphs have twin-width at most~3.
\end{lemma}

\begin{proof}
  Consider $G_t$, the $\ell$-cycle of strict half-graphs of height $t$, on vertex set $\{a^p_1, \ldots, a^p_t~:~p \in [0,\ell-1]\}$.
  By~\cref{obs:monotone} we may prove the stronger statement that $\tww(G'_t) \leqslant 3$ where $G'_t$ is obtained from $G_t$ by adding the red edge $a^p_1a^{p+1 \mod \ell}_1$ for every $p \in [0,\ell-1]$.
  $G'_1$ is a red cycle, and admits a 2-sequence by iteratively contracting the endpoints of any red edge.
  Therefore, we shall just check that the following is a partial 3-sequence from $G'_t$ to $G'_{t-1}$: $(a^0_1,a^0_2), (a^1_1,a^1_2), \ldots, (a^{\ell-1}_1,a^{\ell-1}_2)$.
  Indeed $\{a^p_1,a^p_2\}$ has, in the suggested partial sequence, at most three red neighbors: $a^{p-1 \mod \ell}_1$, $a^{p+1 \mod \ell}_1$, and $a^{p+1 \mod \ell}_2$; the latter two being a single vertex when $p=\ell-1$.
  And all the other vertices have red degree at most~2.
\end{proof}

The next lemma was already invoked in the introduction.
We will use its essence in a more general form in \cref{subsec:tww-bound}, namely that contracting vertices from a module $X$ relative to $Y$ do not create red edges incident to $Y$.  
\begin{lemma}\label{lem:modular}
  Let $G$ be a graph and $\mathcal H = \{H_1, H_2, \ldots, H_\ell\}$ be its modular partition.
  Then, $$tww(G)=\max\{\underset{i \in [\ell]}{\max}~tww(H_i),tww(G/\mathcal H)\}.$$
\end{lemma}
\begin{proof}
  Consider a sequence which contracts each module $H_i$ into a single vertex, followed by a contraction sequence of $G/\mathcal H$ attaining $tww(G/\mathcal H)$.
  Contracting each module $H_i$ creates red edges only within the module $H_i$, thus $G/\mathcal H$ has red degree~0.
  Hence, $tww(G)$ is upper bounded by the twin-widths of $G[H_1], G[H_2], \ldots, G[H_\ell]$, and $G/\mathcal H$.
  The other inequality follows from~\cref{obs:induced-subgraph}, since $G/\mathcal H$ is also an induced subgraph of $G$.
\end{proof}

Finally we will use the following lemma to finish our 4-sequences.
\begin{lemma}\label{lem:subdivided-grid}
Any trigraph whose total graph is a subdivision of a subgraph of a grid has twin-width at most~4.
\end{lemma}
\begin{proof}
By~\cref{obs:monotone}, we may assume that the given trigraph is a subdivision of a complete grid consisting of red edges only. 
By contracting every adjacent pair of vertices one of which has degree~2, we obtain a (complete) grid. 
In a (red) grid on the vertex set $[n] \times [m]$, consider a matching of $n$ leftmost horizontal edges, that is the edges connecting $(i,1)$ and $(i,2)$ for every $i \in [n]$.
Contracting these edges from top to bottom produces intermediate red graphs of red degree at most~4, ending in a red grid graph on $[n] \times [m-1]$.
Finally a path is obtained, which can be contracted into a single vertex while keeping red degree at most~2. 
\end{proof}

\subsection{List of handled problems}



We refer to~\cref{subsec:graph-theory} for the definitions of dominating set, vertex cover, connected set, and independent set.
We will mostly deal with the following three problems, given by order of importance.

\defparproblem{\kds}{A graph $G$ and an integer $k$}{$k$}{Does $G$ have a dominating set of size at most $k$?}

\defparproblem{\kconvc (\textsc{Connected $k$-VC})}{A graph $G$ and an integer $k$}{$k$}{Does $G$ have a  vertex cover of size at most $k$ which induces a connected subgraph of $G$?}

Given a graph $G$ and a capacity function $c: V(G) \rightarrow \mathbb{N}$, a \emph{capacitated vertex cover} $X$ of $G$ is a vertex cover of $G$ which admits a mapping $\rho: E(G) \rightarrow X$ assigning to each vertex $x \in X$ no more edges than its capacity, i.e., $\abs{\rho^{-1}(x)} \leqslant c(x)$ for every $x \in X$.

\defparproblem{\kcapvc (\textsc{Capacitated $k$-VC})}{A graph $G$ with a capacity function $c: V(G) \rightarrow \mathbb{N}$ and an integer $k$}{$k$}{Does $G$ admit a capacitated vertex cover $X$ of size at most $k$?}

The lower bound for \kds will also apply to its \emph{connected} and \emph{total} variants.

\defparproblem{\textsc{Connected $k$-Dominating Set} (\textsc{Connected $k$-DS})}{A graph $G$ and an integer $k$}{$k$}{Does $G$ have a dominating set of size at most $k$ which induces a connected subgraph of $G$?}

\defparproblem{\textsc{Total $k$-Dominating Set}}{A graph $G$ and an integer $k$}{$k$}{Does $G$ admit a set $X \subseteq V(G)$ of size at most $k$ such that every vertex of $G$ has a neighbor in $X$?}

For completeness, we give the definition of the remaining problems that are mentioned at least once in the paper.

\defparproblem{\kis}{A graph $G$ and an integer $k$}{$k$}{Does $G$ have an independent set of size at least $k$?}

\defparproblem{\textsc{Independent $k$-Dominating Set}}{A graph $G$ and an integer $k$}{$k$}{Does $G$ have an independent dominating set of size at most $k$?}

\defparproblem{\textsc{$k$-Path}}{A graph $G$ and an integer $k$}{$k$}{Does $G$ have a path of length at least $k$?}

\defparproblem{\textsc{$k$-Induced Path}}{A graph $G$ and an integer $k$}{$k$}{Does $G$ have an induced path of length at least $k$?}

\defparproblem{\textsc{$k$-Induced Matching}}{A graph $G$ and an integer $k$}{$k$}{Does $G$ have a matching $M$ of size at least $k$ such that there are no edges between endpoints of $M$ other than the edges of $M$?}

\subsection{Kernels or lack thereof}


We often identify a problem to the language made by its positive instances.
For a parameterized problem $\cQ$, a~\emph{kernel} of size bounded by a function $f$ is a polynomial-time reduction $\rho: \Sigma^* \times \mathbb N \to \Sigma^* \times \mathbb N$ such that $(x,k) \in \cQ$ if and only if $\rho(x,k) \in \cQ$, and $|\rho(x,k)| \leqslant f(k)$.
A kernel is said~\emph{linear}, \emph{quadratic}, or \emph{polynomial}, if the function $f$ can be chosen linear, quadratic, or polynomial, respectively.

We recall the framework of OR-cross-compositions~\cite{complexity}, which we will rely on to show the absence of a polynomial kernel in \cref{thm:main}.
An OR-cross-composition is a polynomial reduction that takes $t$ instances of an NP-hard problem $\cL$ and builds an instance of a parameterized problem $\cQ$, such that the OR of the input instances is equivalent to the output instance.
More precisely at least one of the $t$ input instances is positive if and only if the output instance is positive. 
It furthermore allows to set some restrictions on the input NP-hard instances in the form of a polynomial equivalence relation.
This relation can be useful in shaping the input instances in such a way that the composition behaves well.

\begin{definition}\label{def:poly-equivalence}
	A \textit{polynomial equivalence relation} on $\Sigma^*$ is an equivalence relation $\cR$ when
	\begin{compactenum}[(i)]
		\item\label{defit:poly-equivalence:i} for $x, y \in \Sigma^*$, the equivalence $x \cR y$ can be decided in time polynomial in $|x|+|y|$, and
		\item\label{defit:poly-equivalence:ii} $\cR$ restricted to instances of size at most $n$ admits polynomially many equivalence classes.
	\end{compactenum}
\end{definition}

We can now formally define an~\emph{OR-cross-composition}.

\begin{definition}\label{def:or-composition}
  Let $\cL$ be a language, $\cR$ a polynomial equivalence relation on $\Sigma^*$ and $\cQ$ a parameterized problem.
  An \emph{OR-cross-composition} from $\cL$ to $\cQ$ with respect to $\cR$ is an algorithm taking as input $t$ $\cR$-equivalent instances $x_1,...,x_t \in \Sigma^*$, running in time polynomial in $\sum_{j=1}^t |x_j|$, and outputting an instance $(y,N) \in \Sigma \times \mathbb{N}$ such that:
	\begin{enumerate}[(i)]
		\item $N$ is polynomially bounded in $\max_{j \in [t]} |x_j| + \log t$,
		\item $(y,N) \in \cQ$ if and only if there exists some $j$ such that $x_j \in \cL$.
	\end{enumerate}
\end{definition}

We say that $\cL$ cross-composes into $\cQ$, and we sometimes refer to output instance $(y,N)$ as the \emph{composed instance}.
The following result provides the lower bound under {\coNP} $\subseteq$ \NP/poly.

\begin{theorem}\label{thm:kernel-lb}
  If an NP-hard language $\cL$ admits an OR-cross-composition into a parameterized problem $\cQ$, then $\cQ$ does not admit a polynomial kernel unless {\coNP} $\subseteq$ \NP/poly.
\end{theorem}

\subsection{Organization of the rest of the paper}

In~\cref{sec:outline}, we sketch the proof of our main result,~\cref{thm:main}.
In~\cref{sec:tailored}, we show the NP-hardness of \textsc{Minimum Dominating Set} in a customized setting that will be particularly convenient for the subsequent OR-cross-composition.
In~\cref{sec:composition} we present the OR-cross-composition, show its correctness, and prove the twin-width upper bound, thereby establishing~\cref{thm:main}.
In~\cref{sec:positive}, we present two simple $O(k^2)$ kernels for \kconvc and \kcapvc, and a refined $O(k^{1.5})$ kernel for the former problem.
In~\cref{sec:poly}, we show that graphs of twin-width 1 can be efficiently recognized and a 1-sequence with additional property can be produced, and argue that 
this extra property can be used for polynomial-time algorithm for a wide range of problems. 

\section{Outline of the OR-cross-composition for \kds}\label{sec:outline}

Let us start explaining why we should \emph{not} expect a simple OR-composition.
After all, it is only fair to ask for a justification that the kernel lower bound for \kds spans a dozen of pages when the same result for \kis is a side note; even more so, when the former is often more intractable than the latter from a~parameterized complexity standpoint.

The simplest OR-composition is the disjoint union of the input instances.
By~\emph{simple} OR-composition we mean one, like for \kis, straightforwardly based on juxtaposing the instances. 
A standard way to OR-compose $t$ \textsc{Dominating Set}-instances is to have for each instance a ``switch'', that is, one vertex dominating all but one instance.
Then picking the corresponding vertex in the solution, one is left with dominating one chosen instance with a given remaining budget.
This is precisely what we want, but how to ensure that one does not activate two switches?

As we previously observed~\cite{twin-width3}, one can use larger weights for the switches.
However removing the vertex-weights cannot be done without increasing the twin-width.
Another possibility is to force all the budget but one unit (for the switch) within the instances.
This requires, say, $k$ vertices called ``forcers'', each adjacent to a $k$-th fraction of each instance.
Now consider the induced subgraph made by these $k$ vertices, the $t$ switches, and $tk$ vertices of the instances realizing the $tk$ possible neighborhoods toward the former $t+k$ vertices.
The two neighborhoods of every pair of vertices in this graph has a large symmetric difference.
Thus in particular the overall graph has unbounded twin-width.
(Finally known tricks to condense the $t$ switches into $O(\log t)$ vertices do not help, since we want the twin-width to be bounded by an absolute constant.)

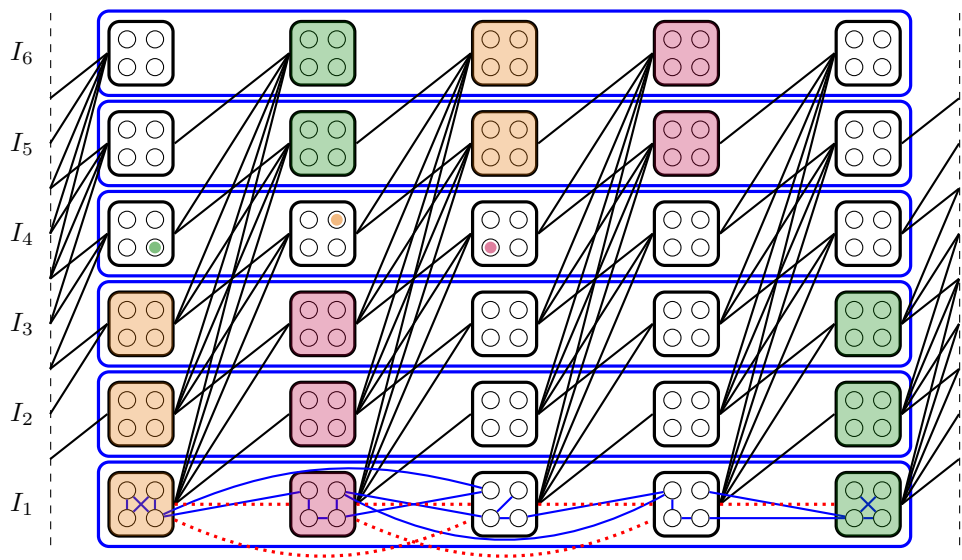
\begin{figure}
  \centering
  \begin{tikzpicture}[scale=.92,vertex/.style={draw,circle,inner sep=0.08cm},picked/.style={fill,circle,opacity=0.5,inner sep=0.055cm}]
    \def\t{6}
    \def\k{5}
    \def\hs{2.6}
    \def\vs{1.3}
    \def\z{0.4}
    \pgfmathtruncatemacro\tm{\t-1}
    \pgfmathtruncatemacro\km{\k-1}
    \pgfmathtruncatemacro\kp{\k+1}
    \foreach \i in {1,...,\t}{
      \begin{scope}[yshift=\i * \vs cm]
        \foreach \j in {1,...,\k}{
          \begin{scope}[xshift=\j * \hs cm]
            \foreach \s/\x/\y in {1/0/0,2/0/1,3/1/0,4/1/1}{
              \node[vertex] (a\i\j\s) at (\x * \z,\y * \z) {} ;
            }
            \node[draw, very thick, rounded corners, fit=(a\i\j1) (a\i\j4)] (b\i\j) {} ;
          \end{scope}
        }
        \node[draw, blue, very thick, rounded corners, fit=(b\i1) (b\i\k)] (t\i) {} ;
        \node at (1.1,0.5 * \z) {$I_\i$} ;
      \end{scope}
    }
    \foreach \zz in {0,\kp}{
     \foreach \i in {1,...,\t}{
      \begin{scope}[yshift=\i * \vs cm, xshift=\zz * \hs cm]
            \foreach \s/\x/\y in {1/0/0,2/0/1,3/1/0,4/1/1}{
              \node[circle,inner sep=0.08cm] (a\i\zz\s) at (\x * \z,\y * \z) {} ;
            }
            \node[very thick, rounded corners, fit=(a\i\zz1) (a\i\zz4)] (b\i\zz) {} ;
          \end{scope}
     }
    }

    \foreach \j [count=\jp from 2] in {1,...,\km}{
      \foreach \i in {1,...,\tm}{
        \pgfmathtruncatemacro\ip{\i+1}
        \foreach \iq in {\ip,...,\t}{
          \draw[thick] (b\i\j.east) -- (b\iq\jp.west) ;
        }
      }
    }
    \foreach \i in {1,...,\tm}{
        \pgfmathtruncatemacro\ip{\i+1}
        \foreach \iq in {\ip,...,\t}{
          \path[draw half paths={thick}{draw=none}] (b\i\k.east) -- (b\iq\kp.west) ;
          \path[draw half paths={draw=none}{thick}] (b\i0.east) -- (b\iq1.west) ;
        }
    }
    \draw[thin,dashed] (0.58 * \hs,0.7 * \vs) -- (0.58 * \hs,\t * \vs + 0.6 * \vs) ;
    \draw[thin,dashed] (0.58 * \hs + \k * \hs,0.7 * \vs) -- (0.58 * \hs + \k * \hs,\t * \vs + 0.6 * \vs) ; 

    \foreach \j/\s/\jj/\ss in {1/3/1/2,1/3/1/4,1/2/1/1,1/1/1/4, 2/3/2/1,2/3/2/4,2/2/2/1, 3/1/3/4,3/1/3/3, 4/1/4/2,4/1/4/3, 5/1/5/3,5/1/5/4,5/2/5/3, 1/3/2/2, 2/4/3/1,2/3/3/2, 3/3/4/2, 4/3/5/1,4/4/5/1}{
      \draw[thick, blue] (a1\j\s) -- (a1\jj\ss) ; 
    }
    \foreach \j/\s/\jj/\ss/\b in {1/3/3/2/20,2/4/4/2/-29}{
      \draw[thick, blue] (a1\j\s) to [bend left = \b] (a1\jj\ss) ; 
    }
    \foreach \j [count = \jp from 2] in {1,...,\km}{
      \draw[very thick,red,dotted] (b1\j) -- (b1\jp) ;
    }
    \foreach \j/\jj/\b in {1/3/-25,2/4/-25}{
      \draw[very thick,red,dotted] (b1\j) to [bend left = \b] (b1\jj) ;
    }
    \node[picked,green!50!black] at (\hs+\z,4 * \vs) {} ;
    \node[picked,orange!90!black] at (2 * \hs+\z,4 * \vs+\z) {} ;
    \node[picked,purple] at (3 * \hs,4 * \vs) {} ;
    \foreach \i in {5,6}{
      \node[fill opacity=0.3, fill=green!50!black, rounded corners, fit=(a\i21) (a\i24)] {} ;
      \node[fill opacity=0.3, fill=orange!90!black, rounded corners, fit=(a\i31) (a\i34)] {} ;
      \node[fill opacity=0.3, fill=purple, rounded corners, fit=(a\i41) (a\i44)] {} ;
    }
    \foreach \i in {1,2,3}{
      \node[fill opacity=0.3, fill=green!50!black, rounded corners, fit=(a\i51) (a\i54)] {} ;
      \node[fill opacity=0.3, fill=orange!90!black, rounded corners, fit=(a\i11) (a\i14)] {} ;
      \node[fill opacity=0.3, fill=purple, rounded corners, fit=(a\i21) (a\i24)] {} ;
    }
  \end{tikzpicture}
  \caption{The overall picture. Instances $I_1, \ldots, I_t$ (here with $t=6$) are in rows, boxed in blue, with their edge also in blue.
    For the sake of legibility, we only represented the edges of $I_1$.
    The red dotted edges are the red edges appearing after contracting every part (boxed in black) into a single vertex.
    Example of what three vertices picked in the first three parts of $I_4$ dominates in the other instances.
    Continuing consistently in $I_4$ would result in ``switching off'' all the other instances, while deviating would leave at least one part ``white'' and not intersected, thus one vertex not dominated.}
  \label{fig:overall-sketch}
\end{figure}

So we need a more elaborate way of selecting one instance among $t$; one, thought primarily to keep the twin-width low.
In the previous attempts, the twin-width was increasing too much because of attachments --switches and forcers-- external to the instances.
We will therefore have instances themselves play these roles.
Say that each instance comes with a partition of its vertex set into $N$ parts, each of which containing a vertex solely adjacent to vertices in its part.
We place the $t$ instances in a $t \times N$ two-dimensional layout, where each instance occupies a ``row,'' while the $j$-th part of all the instances form the $j$-th ``column.''
The switch mechanism is as follows.
Every vertex in the $j$-th part of the $i$-th instance, say $I_i$, dominates the $j-1$-st part of the instances with a smaller index, and the $j+1$-st part of the instances with a larger index.
In other words, we put a strict half-graph over the parts of two consecutive columns.
This is done cylindrically, see~\cref{fig:overall-sketch}.

With that mechanism, a dominating set of a fixed instance $I_i$ (intersecting each of its parts once) is a dominating set of the overall graph.
We skip here the details of the reverse direction, but the use of half-graphs and of vertices whose neighborhood in their instance is confined to their own part (the last ingredient is to have a dummy, edgeless top instance $I_t$) should give a feel for why no other kind of dominating sets of size $N$ can exist.

What about the twin-width bound?
Cycles of half-graphs have bounded twin-width.
So a natural first step is to contract every part of every instance into a single vertex.
Doing so will create some red edges within each row.
To ensure that the red degree remains bounded in this first step, a part should be partially adjacent to only a bounded number of other parts.
In the second step, we contract the cycle of half-graphs row by row.
Thus the red edges of the different instances will progressively stack up.
We need to control the accretion with the red edges of each instance mapping onto a common bounded-degree red graph.
Finally in the third step, we contract the residual red graph.
It should be itself of bounded twin-width, for instance by being planar.

In the next section, we show that \textsc{Minimum Dominating Set} remains NP-hard even when inputs are equipped with a vertex-partition satisfying all the properties that we came across in this outline. 

\section{Tailored NP-hardness for \ds}\label{sec:tailored}

We will show in this section the following hardness result for \ds.
The extra properties that we get compared to existing NP-hardnesses of \ds (even those on planar instances) is crucial for the subsequent OR-cross-composition, as hinted at in the previous section.

The statement of the next two theorems involve what we will call \emph{snaking grids}.
See~\cref{fig:snaking-grid} for an illustration of the \emph{$5 \times 10$ snaking grid}, which has $(3 \cdot (5-1) + 1)(3 \cdot (10-1) + 1)$ vertices.
One may observe that the snaking grids are subdivisions of a wall with some extra isolated vertices.
We will prefer to think of the snaking grid as a spanning subgraph of a (complete) grid, hence the particular embedding of the figure.
The motivation behind the snaking grid will become clear when designing the cross-composition and bounding its twin-width in~\cref{subsec:tww-bound}: it allows to superpose a canonical hamiltonian cycle such that the maximum degree remains~3. 

\begin{figure}[h!]
  \centering
  \resizebox{390pt}{!}{
  \begin{tikzpicture}[scale=.85,
      vertex/.style={draw,thick,rectangle,rounded corners,inner sep=0.22cm}]
    \def\red{red!70!black}
    \def\n{4}
    \def\m{9}
    \pgfmathtruncatemacro\nz{\n - 1}
    \pgfmathtruncatemacro\nn{3 * \n + 1}
    \pgfmathtruncatemacro\mm{3 * \m + 1}
    \pgfmathtruncatemacro\nt{3 * \n}
    \pgfmathtruncatemacro\mt{3 * \m}
    \pgfmathtruncatemacro\np{3 * \n - 2}
    \pgfmathtruncatemacro\mp{3 * \m - 2}
    \def\h{1}
    \def\v{1}
    \foreach \j in {1,...,\mm}{
      \foreach \i in {1,...,\nn}{
        \node[vertex] (b\j-\i) at (\j * \h, \i * \v) {} ;
      }
    }
    \foreach \i in {4,7,...,\nn}{
     \foreach \j in {1,3,...,\mt}{
      \pgfmathtruncatemacro\jp{\j+1}
      \draw[very thick,red] (b\j-\i) -- (b\jp-\i) ;
     }
    }
    \foreach \i in {3,6,...,\nn}{
     \foreach \j in {3,5,...,\mm}{
      \pgfmathtruncatemacro\jm{\j-1}
      \draw[very thick,red] (b\j-\i) -- (b\jm-\i) ;
     }
    }
    \foreach \j in {1,4,...,\mm}{
      \foreach \i [count = \im from 1] in {2,...,\nn}{
        \draw[very thick,red] (b\j-\i) -- (b\j-\im) ;
      }
    }
    \foreach \j [count = \jm from 1] in {2,...,\mm}{
        \draw[very thick,red] (b\j-1) -- (b\jm-1) ;
    }
    \foreach \i in {3,6,...,\nt}{
     \pgfmathtruncatemacro\ip{\i+1} 
     \foreach \j in {2,5,...,\mt,3,6,...,\mt}{
      \pgfmathtruncatemacro\jp{\j+1}
      \draw[very thick,red] (b\j-\i) -- (b\j-\ip) ;
     }
    }
  \end{tikzpicture}
  }
  \caption{The $5 \times 10$ snaking grid.
  }
  \label{fig:snaking-grid}
\end{figure}
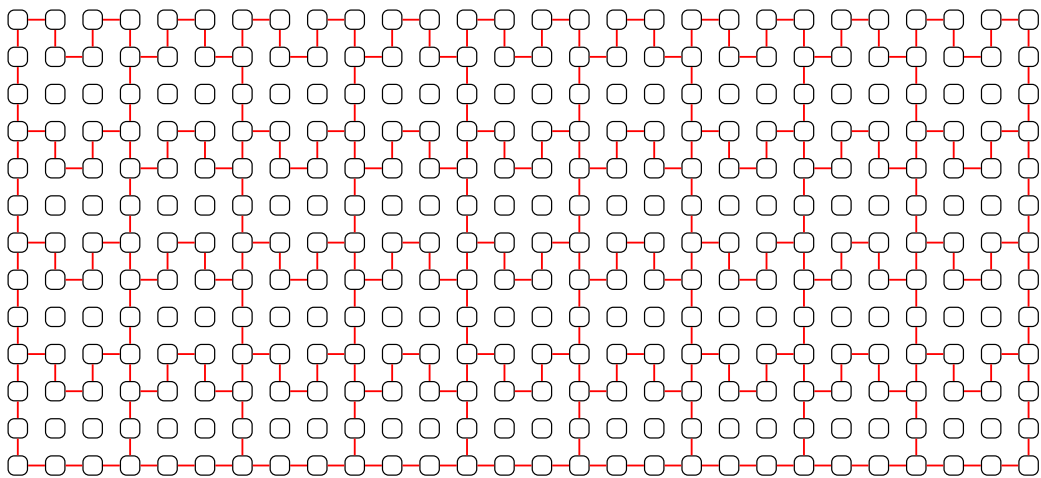

\begin{theorem}\label{thm:tailored-ds}
  \ds remains \NP-hard when its input $(G,N)$ comes with a vertex-partition $\mathcal B = \{B_1, \ldots, B_N\}$, two positive integers $s$ and $t$, and a bijective mapping $\eta$ from $\{B_1, \ldots, B_N\}$ to the vertex set of the $s \times t$ snaking grid such that:
  \begin{compactenum}[(i)]
  \item\label{it1} $G$ has a partial 4-sequence to the quotient trigraph $G/\mathcal B$, 
  \item\label{it2} $G/\mathcal B$ is a spanning subgraph of the $s \times t$ snaking grid, with $t$ even, witnessed by $\eta$, and 
  \item\label{it3} every dominating set of $G$ intersects each $B_i$, for $i \in [N]$.
  \end{compactenum}
\end{theorem}
 
We will obtain~\cref{thm:tailored-ds} as a direct consequence of the following reduction, and the fact that \textsc{Planar 3-SAT} is \NP-hard~\cite{Lichtenstein82} (see~\cref{subsec:planarsat}).
 
\begin{theorem}\label{thm:np-hardness}
  There is a polynomial-time reduction from \textsc{Planar 3-SAT} to \ds that, on $n$-variable $m$-clause \textsc{Planar 3-SAT}-instances $\varphi$ produces \ds-instances $(G,N,\mathcal B= \{B_1, \ldots, B_N\},\eta)$ such that:
  \begin{compactenum}[(i)]
  \item\label{it1} $\mathcal B$ partitions $V(G)$, and $G$ has a partial 4-sequence to the quotient trigraph $G/\mathcal B$, 
  \item\label{it2} $G/\mathcal B$ is a spanning subgraph of the $(m+1) \times n'$ snaking grid, with $n'$ even in $\{n,n+1\}$, witnessed by $\eta$, 
  \item\label{it3} every dominating set of $G$ intersects each $B_i$, for $i \in [N]$, and
  \item\label{it4} $\varphi$ is satisfiable if and only if $G$ has a dominating set of size $N$.
  \end{compactenum}
\end{theorem}

One can observe that conditions $(\ref{it1})$, $(\ref{it2})$, and $(\ref{it3})$ of \cref{thm:tailored-ds,thm:np-hardness} match.
Condition $(\ref{it4})$ in~\cref{thm:np-hardness} only insists that the polynomial-time transformation is a valid reduction.

\subsection{Planar satisfiability}\label{subsec:planarsat}

\textsc{Planar 3-SAT} was introduced by Lichtenstein, who showed its NP-hardness~\cite{Lichtenstein82}, as a convenient starting point to prove the intractability of planar problems.
It is a restriction of \textsc{3-SAT} where the variable/clause incidence graph is planar even if one adds edges between two consecutive variables for a specified ordering of the variables: $x_1, x_2, \ldots, x_n$; i.e., $x_ix_{i+1}$ is an edge (with index $i+1$ taken modulo $n$).
In any \textsc{Planar 3-SAT}-instance $\varphi$ one can partition the clause set $\mathcal C$ into $(\mathcal C^+,\mathcal C^-)$ such that $\mathcal C^+$ and $\mathcal C^-$ each admits a removal ordering, where a \emph{removal ordering} consists of iteratively applying the two following kinds of deletions:
\begin{compactitem}
\item removing a variable which is not present in any remaining clause, or
\item removing a clause on \emph{three consecutive variables} together with the \emph{middle variable},
\end{compactitem}
which ends up with an empty set of clauses.
\emph{Three consecutive variables} means three variables $x_i$, $x_j$, $x_k$, with $i<j<k$ such that $x_{i+1}, x_{i+2}, \ldots, x_{j-1}$ and $x_{j+1}, x_{j+2}, \ldots, x_{k-1}$ have all been removed already.
The middle variable of the clause is $x_j$.
For an example, see~\cref{fig:planar3sat}.
What matters to us is that $\varphi$ can be embedded in an $(|\mathcal C|+1) \times n$ grid (the dashed blue lines in~\cref{fig:planar3sat}) such that every ``event'' (a variable ``wire'' turns or finishes) happens at a grid point.

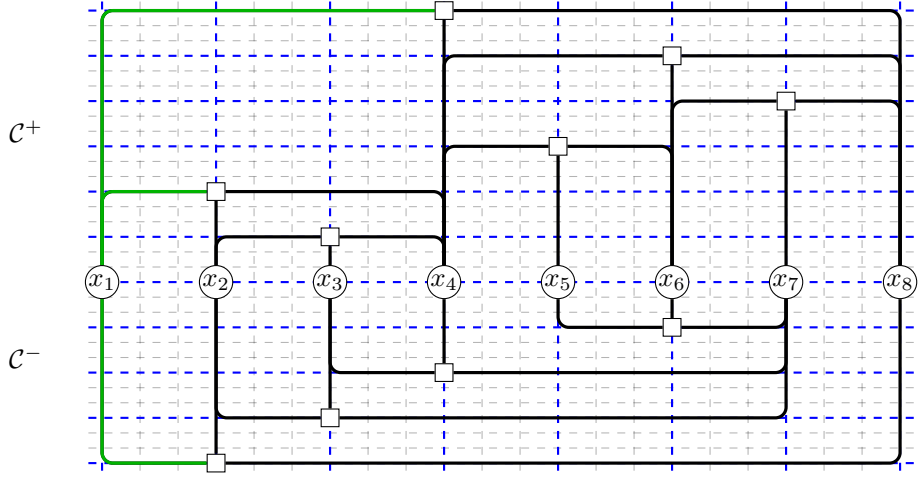
\begin{figure}[h!]
\centering
\begin{tikzpicture}[scale=1]
\def\n{8}
\pgfmathsetmacro{\pn}{\n-1}
\pgfmathsetmacro{\nn}{2 * \n}
\def\s{1.5}
\def\z{0.6}

\foreach \i in {3,5,...,\nn}{
  \draw[dashed,opacity=0.3] (\i * \s / 2 - \s /6,6.2 * \z) -- (\i * \s / 2 - \s/6,-4.2 * \z) ;
  \draw[dashed,opacity=0.3] (\i * \s / 2 + \s /6,6.2 * \z) -- (\i * \s / 2 + \s/6,-4.2 * \z) ;
}
\foreach \j in {-7,-5,...,11}{
  \draw[dashed,opacity=0.3] (\s-0.3 * \z,\j * \z / 2 - \z /6) -- (\n * \s+0.3 * \z,\j * \z / 2 - \z /6) ;
  \draw[dashed,opacity=0.3] (\s-0.3 * \z,\j * \z / 2 + \z /6) -- (\n * \s+0.3 * \z,\j * \z / 2 + \z /6) ;
}
\foreach \i in {1,...,\n}{
  \draw[thick,dashed,blue] (\i * \s,6.2 * \z) -- (\i * \s,-4.2 * \z) ;
}
\foreach \j in {-4,...,6}{
  \draw[thick,dashed,blue] (\s-0.3 * \z,\j * \z) -- (\n * \s+0.3 * \z,\j * \z) ;
}

\foreach \i in {1,...,\n}{
    \node[fill=white,draw,circle,inner sep=0.01cm] (x\i) at (\s * \i,0) {$x_\i$} ;      
}

\foreach \i/\j/\k/\l/\c in {2/3/4/1/red,1/2/4/2/red,4/5/6/3/red,6/7/8/4/red,4/6/8/5/red,1/4/8/6/red, 5/6/7/-1/blue,3/4/7/-2/blue,2/3/7/-3/blue, 1/2/8/-4/blue}{
  \pgfmathtruncatemacro\ll{\l+5} ;
 \node[draw,fill=white] (c\j\ll) at (\s * \j, \l * \z) {} ;
 \draw[very thick,rounded corners] (x\i) -- (\s * \i, \l * \z)  -- (c\j\ll) -- (\s * \k, \l * \z) -- (x\k) ;
 \draw[very thick] (x\j) -- (c\j\ll) ;
}

\begin{scope}[very thick, green!70!black, rounded corners]
\draw (x1) -- (\s, 6 * \z) -- (c411) ;
\draw (x1) -- (\s, 2 * \z) -- (c27) ;
\draw (x1) -- (\s, -4 * \z) -- (c21) ;
\end{scope}

\node at (0.5,2) {$\mathcal C^+$} ;
\node at (0.5,-1) {$\mathcal C^-$} ; 

\end{tikzpicture}
\caption{The bipartition $(\mathcal C^+,\mathcal C^-)$ of an $n$-variable $m$-clause \textsc{Planar 3-SAT}-instance, with $n=8$ and $m=10$.
  The square vertices represent the clauses.
  All the degree-3 vertices of the solid black (and green) graph lie at intersections of the $(m+1) \times n$ grid represented with the dashed blue lines.
  The $(3m+1) \times (3n-2)$ grid made by the dashed blue and gray lines will embed the snaking grid.
  The ``wire'' of variable $x_1$ is highlighted in green.
}
\label{fig:planar3sat}
\end{figure}

The reduction to establish~\cref{thm:np-hardness} will be a rather transparent substitution of local gadgets on grid points of~\cref{fig:planar3sat}.
We need a variable gadget (to initiate a Boolean choice), a~propagation gadget (``wires'' that may split), and a clause gadget (where the ``wires'' of three variables meet at a grid point).
In addition, we will design a very simple dummy gadget, used to fill all the unoccupied grid points.
We start with a description of the variable and propagation gadgets.

\subsection{Variable and propagation gadgets}

We distinguish the \emph{initial variable gadget} and the \emph{(regular) variable gadget}.
The initial variable gadget is simply a triangle, one vertex of which corresponds to setting the variable to true (marked $\top$ in the figures), a second vertex, to setting it to false (marked $\bot$), and a third vertex which will remain of degree 2.
These three vertices form one part, say, $B_\ell$ of $\mathcal B$.
We denote them by $\top, \bot,$ and $d$, and use the subscript of the part of $\mathcal B$ they belong to whenever we want to be specific (so $\top_\ell, \bot_\ell,$ and $d_\ell$). 

The \emph{regular variable gadget} (or \emph{variable gadget} for short) is the \emph{bull} graph. 
More precisely, it is an initial variable gadget where one adds one pendant neighbor $t$ to $\top$ and one pendant neighbor $f$ to $\bot$.
Again these five vertices form one part of $\mathcal B$, say $B_\ell$, that we may denote $\top_\ell, \bot_\ell, d_\ell, t_\ell,$ and $f_\ell$.
For both the initial and regular variable gadgets, the \emph{index} of the gadget is simply $\ell$.

The propagation gadget from a (possibly initial) variable gadget of index $\ell$ to a regular variable gadget of index $\ell'$ consists of the two edges $\top_\ell f_{\ell'}$ and $\bot_\ell t_{\ell'}$.  The ``propagation'' starts at an initial variable gadget and extends in two directions.
It may also ``split,'' that is, a same regular variable gadget can propagate to two variable gadgets, one horizontally and one vertically.
See~\cref{fig:propagation} for a depiction of these gadgets.

\begin{figure}[h!]
\centering
\begin{tikzpicture}[scale=.8]
  \def\s{2}
  \begin{scope}[xshift=4 * \s cm]
      \foreach \l/\x/\y/\h in {t/0/0/\top,f/0/2/\bot,z/0.3/1/}{
        \node[draw,circle,inner sep=0.05cm] (\l4) at (\x,\y) {$\h$} ;
      }
      \draw (t4) -- (z4) -- (f4) -- (t4) ;
      \node[draw,very thick,rounded corners,fill,fill opacity=0.08,fit=(t4) (f4) (z4)] (B4) {} ;
  \end{scope}
  
  \foreach \i/\m in {1/-1,2/-1,3/-1,5/1,6/1,7/1}{
    \begin{scope}[xshift=\i * \s cm]
      \foreach \l/\x/\y/\h in {t/0/0/\top,f/0/2/\bot,d/-0.35/0.7/,e/-0.35/1.3/,z/0.3/1/}{
        \node[draw,circle,inner sep=0.05cm] (\l\i) at (\x * \m,\y) {$\h$} ;
      }
      \draw (t\i) -- (z\i) -- (f\i) -- (t\i) -- (d\i) ;
      \draw (f\i) -- (e\i) ;
      \node[draw,very thick,rounded corners,fit=(t\i) (f\i) (e\i) (d\i) (z\i)] (B\i) {} ;
    \end{scope}
  }
  \begin{scope}[xshift=6 * \s cm]
      \foreach \l/\x/\y/\h in {t/0/0/\top,f/0/2/\bot,d/-0.35/0.7/,e/-0.35/1.3/,z/0.3/1/}{
        \node[draw,circle,inner sep=0.05cm] (\l8) at (\x,\y+3.4) {$\h$} ;
      }
      \draw (t8) -- (z8) -- (f8) -- (t8) -- (d8) ;
      \draw (f8) -- (e8) ;
      \node[draw,very thick,rounded corners,fit=(t8) (f8) (e8) (d8) (z8)] (B8) {} ;
  \end{scope}
  \draw (t6) to [bend left=43] (e8) ;
  \draw (f6) to [bend left=25] (d8) ;
  
  \foreach \i/\j in {t4/e5,t4/e3,f4/d5,f4/d3}{
    \draw (\i) -- (\j) ;
  }
  \foreach \i [count = \j from 2] in {1,2}{
    \draw (t\j) -- (e\i) ;
    \draw (f\j) -- (d\i) ;
  }
  \foreach \i [count = \j from 6] in {5,6}{
    \draw (t\i) -- (e\j) ;
    \draw (f\i) -- (d\j) ;
  }
  \draw[-stealth] (1,1) --++(-0.8,0) ;
  \draw[-stealth] (15,1) --++(0.8,0) ;
  \draw[-stealth] (12,6.15) --++(0,0.8) ;
\end{tikzpicture}
\caption{Initial variable gadget (slightly shaded box in the middle), seven (regular) variable gadgets, and the propagation gadgets (edges in between variable gadgets), allowing the ``wire'' to split. 
  The thick boxes represent the partition $\mathcal B$, each box corresponding to an initial or regular variable gadget.
  The figure should be rotated 90 degrees to match~\cref{fig:planar3sat} where the propagation from the initial variable gadgets occurs vertically.
  }
\label{fig:propagation}
\end{figure}

For every variable $x_i$ of $\varphi$, we have a unique initial variable gadget, and several regular variable gadgets connected to it via a chain of propagations. 
We call \emph{wire} of $x_i$, the graph induced by all these variable gadgets (the initial variable gadget included).
The \emph{digraph of a wire} (or, by extension, of its corresponding variable) is the directed graph whose vertices are the variable gadgets of the wire, and with an arc from a variable gadget of index $\ell$ to one of index $\ell'$ if there is a propagation gadget from the former to the latter.
Following~\cref{fig:planar3sat} the digraph of every wire is an out-tree of maximum out-degree 2, rooted at the initial variable gadget.
The leaves of this tree are adjacent to a grid point where we will place a clause gadget.

The following lemma ensures the expected behavior of a choice propagation.
The rest of the construction (mainly the insertion of the clause gadgets) will not change the properties used in its proof, so we show it here.
\begin{lemma}\label{lem:propagation}
  Let $H$ be the wire of variable $x$, and let $\vec T$ be the out-tree, digraph of $H$.
  The only two dominating sets of $H$ of size (at most) $|V(\vec T)|$ consists of picking $\top$ in every gadget, or of picking $\bot$ in every gadget.
\end{lemma}  
\begin{proof}
  Every variable gadget has a vertex $d$ only adjacent to vertices within the gadget, so a dominating set of size $|V(\vec T)|$ has to intersect each gadget exactly once.
  Since the digraph of $H$ is supposed to be an out-tree, and the vertices $t_{\ell'}$ and $f_{\ell'}$ have their neighborhood included in the gadget of index $\ell'$ and its unique in-neighbor (in $\vec T$), of index, say, $\ell$, it suffices to observe that the only pairs of vertices dominating all the vertices in these two gadgets are $\{\top_\ell,\top_{\ell'}\}$ and $\{\bot_\ell,\bot_{\ell'}\}$. 
\end{proof}
In the statement of~\cref{lem:propagation}, the former choice \emph{sets $x$ to true}, the latter, \emph{sets $x$ to false}.
  
\subsection{Clause and dummy gadgets}

The \emph{clause gadget} of $C_j$ consists of two vertices $c_j, z_j$, where $c_j$ is linked to its three literals in the corresponding wire ends, and $z_j$ is isolated.
The set $\{c_j,z_j\}$ is added to $\mathcal B$.
See~\cref{fig:clause} for an illustration.

\begin{figure}[h!]
\centering
\begin{tikzpicture}[scale=.8]
  \foreach \i/\xs/\ys/\r/\m in {1/0/0/0/1,2/3.4/-1.2/90/1,3/4.8/0/0/-1}{
    \begin{scope}[xshift=\xs cm,yshift=\ys cm,rotate=\r]
      \foreach \l/\x/\y/\h in {t/0/0/\top,f/0/2/\bot,d/-0.35/0.7/,e/-0.35/1.3/,z/0.3/1/}{
        \node[draw,circle,inner sep=0.05cm] (\l\i) at (\x * \m,\y) {$\h$} ;
      }
      \draw (t\i) -- (z\i) -- (f\i) -- (t\i) -- (d\i) ;
      \draw (f\i) -- (e\i) ;
      \node[draw,very thick,rounded corners,fit=(t\i) (f\i) (e\i) (d\i) (z\i)] (B\i) {} ;
    \end{scope}
  }
  \node[draw,circle,inner sep=0.05cm] (c) at (2.4,1) {$c_j$} ;
  \node[draw,circle,inner sep=0.05cm] (z) at (2.4,2) {$z_j$} ;
  \node[draw,very thick,rounded corners,fit=(c) (z)] (C) {} ;
  \foreach \i in {t1,f2,f3}{
    \draw (c) -- (\i) ;
  }

  \draw[-stealth] (-2.2,1) to node[above] {$x_{a_1}$} (-1.2,1) ;
  \draw[-stealth] (7,1) to node[above] {$x_{a_3}$} (6,1) ;
  \draw[-stealth] (2.4,-3.3) to node[left] {$x_{a_2}$} (2.4,-2.3) ;
\end{tikzpicture}
\caption{Three variable wires converging to a clause gadget $C_j = x_{a_1} \lor \neg x_{a_2} \lor \neg x_{a_3}$.
  Again, the thick boxes represent the partition $\mathcal B$.
  In fact, due to the horizontal snaking, either one of $x_{a_1}, x_{a_3}$ would appear above the clause gadget (see~\cref{fig:clause-gadget-snaking}).}
\label{fig:clause}
\end{figure}

A \emph{dummy gadget} is simply an isolated vertex $z$, which makes its own singleton part $\{z\}$ in $\mathcal B$.
Thus a dummy gadget is a clause gadget minus the vertex $c_j$.
The isolated vertices $z_j$'s in the clause gadgets, and the dummy gadgets may seem artificial.
Their purpose is to fulfill conditions $(\ref{it2})$ and $(\ref{it3})$.
Besides isolated vertices will not remain isolated in the subsequent OR-cross-composition.

\subsection{Construction}

At this point, the construction is probably clear.
Let $n$ and $m$ be the number of variables and clauses, respectively, of the \textsc{Planar 3-SAT}-instance $\varphi$.
We assume that $n$ is even (if not, we add a dummy variable not appearing in any clause).
We put one initial variable gadget at each grid point (of the finer $(3m+1) \times (3n-2)$ grid) in~\cref{fig:planar3sat} with a circled vertex ($n$ grid points in total).
We place a regular variable gadget at every grid point covered by the solid black (or green) lines, and at every grid point just below any horizontal segment of a wire (to account for the snaking), except for those on the bottommost row.
We orient the propagation such that the digraph of each of the $n$ wires is an out-tree rooted at its initial variable gadget.
For every clause of $\mathcal C^-$, we place a clause gadget at the corresponding grid point occupied by a square vertex in~\cref{fig:planar3sat}, and link it accordingly to its literals.
For every clause of $\mathcal C^+$, the clause gadget is instead placed at the grid point just below the corresponding square vertex. 
This is again because of the horizontal snaking (see~\cref{fig:clause-gadget-snaking}).
Finally we put a dummy gadget at every unoccupied grid point.
This finishes the construction of the graph $G$, and its vertex-partition $\mathcal B$.
We set $N$ to $|\mathcal B|$, that is, $(3m+1) \times (3n-2)$.
The embedding witness $\eta$ is implicit from the construction. 

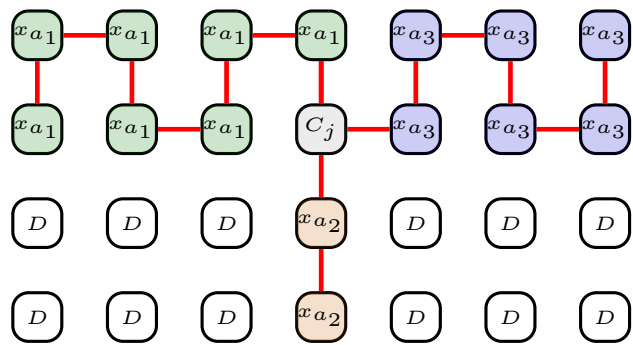
\begin{figure}[h!]
  \centering
  \resizebox{250pt}{!}{
  \begin{tikzpicture}[scale=.85,
      vertex/.style={draw,thick,rectangle,rounded corners,inner sep=0.22cm},
      filled/.style={fill opacity=0.15,fill=white!50!black,rectangle,rounded corners,inner sep=0.22cm},
      filledo/.style={fill opacity=0.2,fill=green!50!black,rectangle,rounded corners,inner sep=0.22cm},
      filledtw/.style={fill opacity=0.2,fill=orange!80!black,rectangle,rounded corners,inner sep=0.22cm},
      filledth/.style={fill opacity=0.2,fill=blue!80!black,rectangle,rounded corners,inner sep=0.22cm}]
    \def\red{red!70!black}
    \def\n{1}
    \def\m{2}
    \pgfmathtruncatemacro\nz{\n - 1}
    \pgfmathtruncatemacro\nn{3 * \n + 1}
    \pgfmathtruncatemacro\mm{3 * \m + 1}
    \pgfmathtruncatemacro\nt{3 * \n}
    \pgfmathtruncatemacro\mt{3 * \m}
    \pgfmathtruncatemacro\np{3 * \n - 2}
    \pgfmathtruncatemacro\mp{3 * \m - 2}
    \def\h{1}
    \def\v{1}
    \foreach \j in {1,...,\mm}{
      \foreach \i in {1,...,\nn}{
        \node[vertex] (b\j-\i) at (\j * \h, \i * \v) {} ;
      }
    }
    \foreach \i in {4}{
     \foreach \j in {1,3,...,\mt}{
      \pgfmathtruncatemacro\jp{\j+1}
      \draw[very thick,red] (b\j-\i) -- (b\jp-\i) ;
     }
    }
    \foreach \i in {3}{
     \foreach \j in {3,5,...,\mm}{
      \pgfmathtruncatemacro\jm{\j-1}
      \draw[very thick,red] (b\j-\i) -- (b\jm-\i) ;
     }
    }
    \foreach \j in {4}{
      \foreach \i [count = \im from 1] in {2,...,\nn}{
        \draw[very thick,red] (b\j-\i) -- (b\j-\im) ;
      }
    }
    \foreach \i in {3}{
     \pgfmathtruncatemacro\ip{\i+1} 
     \foreach \j in {1,...,\mm}{
      \draw[very thick,red] (b\j-\i) -- (b\j-\ip) ;
     }
    }
    \foreach \i/\j/\t/\f in {3/1/{x_{a_1}}/filledo, 4/1/{x_{a_1}}/filledo, 3/2/{x_{a_1}}/filledo, 4/2/{x_{a_1}}/filledo, 3/3/{x_{a_1}}/filledo, 4/3/{x_{a_1}}/filledo, 4/4/{x_{a_1}}/filledo, 1/4/{x_{a_2}}/filledtw, 2/4/{x_{a_2}}/filledtw, 3/5/{x_{a_3}}/filledth,4/5/{x_{a_3}}/filledth,3/6/{x_{a_3}}/filledth,4/6/{x_{a_3}}/filledth,3/7/{x_{a_3}}/filledth,4/7/{x_{a_3}}/filledth, 3/4/{C_j}/filled, 1/1/D/,1/2/D/,1/3/D/,1/5/D/,1/6/D/,1/7/D/,2/1/D/,2/2/D/,2/3/D/,2/5/D/,2/6/D/,2/7/D/}{
      \node[\f] at (\j * \h, \i * \v) {} ;
      \node at (\j * \h, \i * \v) {\tiny{$\t$}} ;
    }
  \end{tikzpicture}
  }
  \caption{The vicinity of a clause gadget, for a clause $C_j \in \mathcal C^+$.
    Red edges materialize the edges in between two parts: propagation and edges incident to the clause gadget.
    Due to the horizontal snaking, the clause gadget is placed just below its position in~\cref{fig:planar3sat}.
    Dummy gadgets are marked with a $D$.}
  \label{fig:clause-gadget-snaking}
\end{figure}

\begin{proof}[Proof of \cref{thm:np-hardness}]
  We now check the conditions $(\ref{it1})$ to $(\ref{it4})$.
  Let us start with the correctness of the reduction, namely $(\ref{it4})$.

  Assume $\varphi$ admits a satisfying assignment $A$, and \emph{sets each variable to the value given by A} in each wire of $G$.
  Complete this solution by including all the isolated vertices of $G$.
  By design this yields a set $S$ of size $N$.
  By~\cref{lem:propagation}, every vertex in a wire is dominated.
  The isolated vertices are dominated by themselves, and every vertex $c_j$ in a clause gadget is dominated since $A$ is a satisfying assignment.
  Thus $S$ is a dominating set of $G$ of size $N$ (intersecting every part of $\mathcal B$ exactly once).

  Conversely, assume that $G$ has a dominating set $S$ of size $N$.
  As every gadget contains a vertex whose neighborhood is included in its own part of $\mathcal B$ (namely $d$, and the isolated vertices), $S$ should intersect every part of $\mathcal B$ exactly once.
  In the dummy and clause gadgets, the picked vertex should thus be the isolated vertex.
  By~\cref{lem:propagation}, in each wire a consistent choice (of taking only $\top$ or only $\bot$) has to be made.
  Let $A$ be the corresponding truth assignment.
  As $S$ dominates each vertex $c_j$ in the clause gadgets, it implies that $A$ is a satisfying assignment, and $\phi$ is satisfiable.

  \textbf{Condition $(\ref{it1})$.}
  By construction $\mathcal B$ partitions $V(G)$.
  We want to propose a partial 4-sequence from $G$ to $G/\mathcal B$.
  In other words, we wish to contract every part of $\mathcal B$ into a single vertex, such that the red degree remains at most~4.

  In any order, we contract the bull graph of each regular variable gadget adjacent to only two other variable gadgets in the following way: $(\top, d)$, $(\{\top, d\}, t)$, $(\{\top, d, t\}, \bot)$,  $(\{\top, d, t, \bot\}, f)$.
  Internally this creates (at any moment) a single red edge.
  There can be up to four red edges between this part and another part of $\mathcal B$ (since there are at most four leaving black edges).
  Observe however that these four edges can be incident to a single vertex only after the last internal contraction $(\{\top, d, t, \bot\}, f)$.
  Thus the red degree of these vertices remains at most~4.

  Meanwhile the red degree of vertices not falling in that category is bounded by~3, in clause gadgets, and by~2, in other variable gadgets.
  We can then contract the initial variable gadget and the clause gadgets (there is nothing to contract in dummy gadgets).
  Regular variable gadgets adjacent to three parts of $\mathcal B$ are now adjacent to three vertices. 
  Indeed there is no pair of adjacent such gadgets.
  Thus the above suggested contraction of the bull keeps red degree at most~4.

  \textbf{Condition $(\ref{it2})$.}
  By design, $G/\mathcal B$ is a subgraph of the $(m+1) \times n$ snaking grid, witnessed by $\eta$.
  It is spanning by introduction of the dummy gadgets.

  \textbf{Condition $(\ref{it3})$.}
  This was observed when checking the correctness of the reduction.
\end{proof}
As previously mentioned, this directly implies~\cref{thm:tailored-ds}.

For \textsc{Connected $k$-Dominating Set} and \textsc{Total $k$-Dominating Set} on bounded twin-width graphs (with an upper bound possibly larger than~4), one does not need the snaking grid nor the dummy gadgets.
In the propagation between two variable gadgets, one can add an edge between the two $\top$ vertices, and an edge between the two $\bot$ vertices.
In the clause gadgets, one can add a neighbor $y_j$ to $z_j$, adjacent to the six vertices $\top$ and $\bot$ in the three incident variable gadgets.

\section{OR-cross-composition}\label{sec:composition}

We now describe the cross-composition from the NP-hard \ds restricted as in \cref{thm:tailored-ds} to \kds.
We use the polynomial equivalence $\cR$ to partition all well-formed instances for the restricted \kds (those satisfying \cref{thm:tailored-ds}) with respect to the given parameter $N$ and dimensions $(p,q)$ of the corresponding snaking grid.
All ill-formed instances will then make a single class.
Relation $\cR$ satisfies the conditions of \cref{def:poly-equivalence}.
For any two well-formed instances $(I_i,N_i,\cB_i,p_i,q_i,\eta_i),(I_{\ell},N_{\ell},\cB_{\ell},p_{\ell},q_{\ell},\eta_{\ell})$, we can check in polynomial time that $N_i = N_{\ell}$ and $(p_i,q_i) = (p_{\ell},q_{\ell})$, yielding $(\ref{defit:poly-equivalence:i})$. Now, choosing some encoding such that all well-formed instances of $\Sigma^{\leq n}$ have at most $n$ vertices, their parameter and snaking grid dimensions must also be bounded by $n$. Then, the number of equivalence classes on $\Sigma^{\leq n}$ accounting for the malformed instances is at most $n^3+1$, yielding $(\ref{defit:poly-equivalence:ii})$.

The conditions for our composition being set, consider $t$ instances of the restricted \ds, equivalent with respect to $\cR$.
If the instances are ill-formed we output an ill-formed instance of \kds.
If not, their equivalence yields common even parameter $N$ and snaking grid dimensions $(p,q)$, letting us consider them as $(I_i,N,\cB_i,p,q,\eta_i)_{i \in [t]}$.
We will construct a \kds instance $(H,N)$, with the same parameter, admitting a solution if and only if at least one input instance $(I_i,N,\cB_i,p,q,\eta_i)$ admits a solution for the restricted \kds.
Before composing the input graphs, we introduce a dummy instance in the form of graph $I_{t+1}$ serving to ensure that any valid $(H,N)$ further admits a solution picking vertices in each column.
$I_{t+1}$ is an independent set of size $2N$ on which we partition $V(I_{t+1})$ through $\cB_{i+1}$ into $N$ classes of exactly two vertices. Note that since $I_{t+1} / \cB_{t+1}$ is an independent set, it is a spanning subgraph of the $p \times q$ snaking grid as witnessed by any bijective $\eta_{t+1}$ onto the latter.

We first show how to order the partition classes of each instance in the same way with respect to their mapping onto the snaking grid. This ordering will follow a fictitious hamiltonian cycle $(y_1,...,y_N)$ on the $p \times q$ snaking grid, which we now describe.
From the snaking grid, we take all edges of the first row, and all edges of the first and last columns.
Then, from the complete underlying grid, we take all but the first edge of each column and we complete the cycle by taking all edges between two degree one vertices in the so-built union of paths (see the darker red cycles in~\cref{fig:contraction-order,fig:layers}).
Referring to the partition of instance $i \in [t+1]$ as $\mathcal{B}_i = \{B_{i,1},...,B_{i,N}\}$, we can assume up to the reordering above that $\eta_i(B_{i,j}) = y_j$.
 
Now, considering all instances over $H$, a representation of the construction that follows is given in \cref{fig:overall-sketch}. It will be useful to consider the instances in a grid such that $B_{i,j}$ is the cell in the $i$-th row and $j$-th column, and we will use the term partition class or cell interchangeably. We can then see instance $I_i$ as row $i$, and define regular instance columns, omitting the dummy instance, as $C_j = \bigcup_{i \in [t]} B_{i,j}$ for $j \in [N]$.

\medskip

\textbf{Construction.} We start building our composed graph $H$ as the union of all instances $(I_i)_{i \in [t+1]}$, that is, $V(H) = \bigcup_{i \in [t+1]} V(I_i)$ and $E(I_i) \subseteq E(H)$ for $i \in [t+1]$.
Then, our cross-composition proceeds by adding a cycle of strict half-graphs over columns $(C_j)_{j \in [N+1]}$: for $i \in [t+1], j \in [N]$, $B_{i,j}$ forms a biclique with $\bigcup_{i < \ell \leq t+1 } B_{\ell,j+1}$ (accounting for indices $j$ modulo $N$).
Notice then that the only edges added above lie between columns $C_j,C_{j'}$ with $|j'-j|=1$, so any edge between two columns differing by at least two indices is an edge of $I_i$. Each instance class $B_{t+1,j}$ is then adjacent exactly to $C_{j-1}$.
Having ordered the classes of each instance in the same way with respect to their mapping onto the snaking grid, column $C_j$ consists of homologous vertices, all in the same position on their respective grids, see \cref{fig:layers}. Then, the cycle of half-graphs follows the darker red fictitious hamiltonian cycles mapping to $(y_1,...,y_N)$.

The composed $k$-\textsc{Dominating Set} is then $(H,N)$, which we can construct polynomially in $\sum_{i=1}^t |I_i|$. Indeed, $N$ is bounded by any $|I_i|$ so the dummy instance is built linearly, then graph union and addition of the bicliques takes time polynomial in $\sum_{i=1}^t |I_i|$. 
Moreover, the reduction directly satisfies condition $(i)$ of \cref{def:or-composition} since the output parameter $N$ is bounded by any $|I_i|$.

\subsection{The overall construction has twin-width at most 4}\label{subsec:tww-bound}

Consider the OR-cross-composition of $t$ instances $(I_i,N,\cB_i,p,q,\eta_i)_{i \in [t]}$ of the restricted \kds described \cref{sec:composition}, letting $(H,N)$ be the composed instance.
Towards showing that the twin-width of $H$ is bounded by $4$, we first show that there is a partial 4-sequence contracting each $B_{i,j}$ for $i \in [t+1], j \in [N]$ into a single vertex in $H$.
After these contractions, each instance $I_i / \cB_i$ can be considered as a red snaking grid augmented by a hamiltonian path (\cref{fig:contraction-order}).
We call~\emph{augmented snaking grid} the graph obtained by adding the (darker red) hamiltonian cycle to the snaking grid. 
The cycle of strict half-graphs added between the instances will then follow the hamiltonian cycles of each augmented snaking grid.
Finally, we show that there is a partial 4-sequence contracting these $t+1$ grids into a single one, which is of twin-width four.

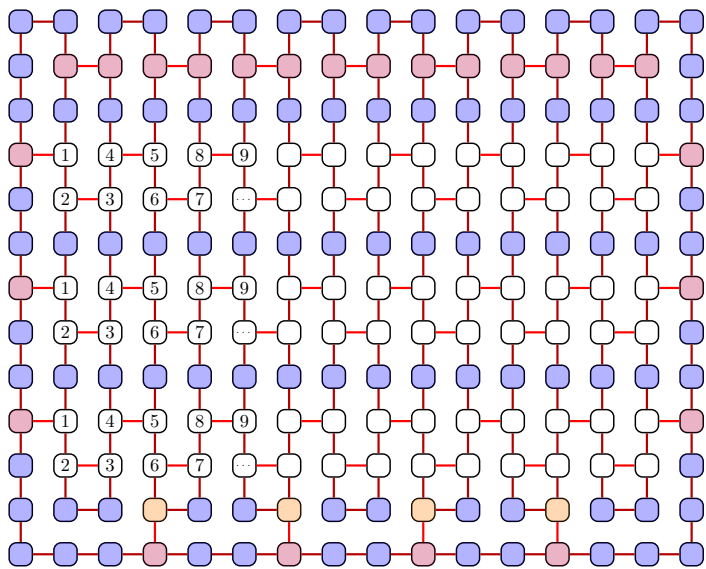
\begin{figure}[h!]
  \centering
  \resizebox{265pt}{!}{
  \begin{tikzpicture}[scale=.85,
      vertex/.style={draw,thick,rectangle,rounded corners,inner sep=0.22cm},
      filled/.style={fill,opacity=0.3,rectangle,rounded corners,inner sep=0.22cm},
      filledb/.style={fill=blue,opacity=0.3,rectangle,rounded corners,inner sep=0.22cm},
      filledp/.style={fill=purple,opacity=0.3,rectangle,rounded corners,inner sep=0.22cm},
      filledr/.style={fill=orange,opacity=0.3,rectangle,rounded corners,inner sep=0.22cm}]
    \def\red{red!70!black}
    \def\n{4}
    \def\m{5}
    \pgfmathtruncatemacro\nz{\n - 1}
    \pgfmathtruncatemacro\nn{3 * \n + 1}
    \pgfmathtruncatemacro\mm{3 * \m + 1}
    \pgfmathtruncatemacro\nt{3 * \n}
    \pgfmathtruncatemacro\mt{3 * \m}
    \pgfmathtruncatemacro\np{3 * \n - 2}
    \pgfmathtruncatemacro\mp{3 * \m - 2}
    \def\h{1}
    \def\v{1}
    \foreach \j in {1,...,\mm}{
      \foreach \i in {1,...,\nn}{
        \node[vertex] (b\j-\i) at (\j * \h, \i * \v) {} ;
      }
    }
    \foreach \j in {1,...,\mm}{
      \foreach \i [count = \ip from 3] in {2,...,\nt}{
        \draw[very thick,\red] (b\j-\i) -- (b\j-\ip) ;
      }
    }
    \foreach \j in {1,...,\mt}{
      \pgfmathtruncatemacro\jp{\j+1}
      \draw[very thick,\red] (b\j-1) -- (b\jp-1) ;
    }
    \foreach \j in {1,3,...,\mt}{
      \pgfmathtruncatemacro\jp{\j+1}
      \draw[very thick,\red] (b\j-\nn) -- (b\jp-\nn) ;
    }
    \foreach \j in {3,5,...,\mt}{
      \pgfmathtruncatemacro\jm{\j-1}
      \draw[very thick,\red] (b\j-2) -- (b\jm-2) ;
    }
    \draw[very thick,\red] (b1-1) -- (b1-2) ;
    \draw[very thick,\red] (b\mm-1) -- (b\mm-2) ;

    \foreach \i in {4,7,...,\np}{
     \foreach \j in {1,3,...,\mt}{
      \pgfmathtruncatemacro\jp{\j+1}
      \draw[very thick,red] (b\j-\i) -- (b\jp-\i) ;
     }
    }
    \foreach \i in {3,6,...,\nt}{
     \foreach \j in {3,5,...,\mm}{
      \pgfmathtruncatemacro\jm{\j-1}
      \draw[very thick,red] (b\j-\i) -- (b\jm-\i) ;
     }
    }
    \foreach \j in {4,7,...,\mp}{ 
    \draw[very thick,red] (b\j-1) -- (b\j-2) ;
    }

    \foreach \j in {2,...,\mt}{
    \foreach \i in {5,8,...,\nn,\nn}{
    \node[filledb] at (\j * \h, \i * \v) {} ;
      }
    }
    \foreach \j in {1,\mm}{
    \foreach \i in {2,5,...,\nn,\nn,1}{
      \node[filledb] at (\j * \h, \i * \v) {} ;
    }
    \foreach \i in {3,6,...,\nn}{
    \node[filledb] at (\j * \h, \i * \v) {} ;
      }
    }
    \foreach \j in {2,5,...,\mm, 3,6,...,\mm}{
      \node[filledb] at (\j * \h, \v) {} ;
      \node[filledb] at (\j * \h, 2 * \v) {} ;
    }
    \foreach \j in {2,...,\mt}{
      \node[filledp] at (\j * \h, \nt * \v) {} ;
    }
    \foreach \j in {1,\mm}{
    \foreach \i in {4,7,...,\np}{
      \node[filledp] at (\j * \h, \i * \v) {} ;
    }
    }
    \foreach \j in {4,7,...,\mp}{
      \node[filledp] at (\j * \h, \v) {} ;
    }
    \foreach \j in {4,7,...,\mp}{
      \node[filledr] at (\j * \h, 2 * \v) {} ;
    }
    \foreach \zz [count = \z from 0] in {1,...,\nz}{
    \foreach \j/\i/\c in {2/4/1,2/3/2,3/3/3,3/4/4,4/4/5,4/3/6,5/3/7,5/4/8,6/4/9,6/3/\tiny{\dots}}{
      \node at (\j * \h, \i * \v + 3 * \z * \v) {\c} ;
    }
    }
  \end{tikzpicture}
  }
  \caption{The maximal set of red edges, with the hamiltonian cycle in darker red, and the edges of the snaking grid that are not already in the hamiltonian cycle in lighter red.
    A partial 4-sequence goes as follows.
    First contract every blue vertex with their homologous in the next layer, in any order.
    Then similarly contract every purple vertex, followed by every orange vertex.
    Finally follow the order indicated by positive integers (in increasing value).
    This way a vertex contracted with the next layer has at most two neighbors, or at most one non-contracted neighbor.
  }
  \label{fig:contraction-order}
\end{figure}

\begin{figure}[h!]
  \centering
  \resizebox{280pt}{!}{
  \begin{tikzpicture}[scale=1.2,
      vertex/.style={draw,thick,circle,inner sep=0.15cm},
      filled/.style={fill,opacity=0.3,circle,inner sep=0.15cm},
      filledb/.style={fill=blue!30!white,circle,inner sep=0.15cm},
      filledp/.style={fill=purple!30!white,circle,inner sep=0.15cm},
      filledr/.style={fill=orange!30!white,circle,inner sep=0.15cm}]
    \def\red{red!70!black}
    \def\n{2}
    \def\m{3}
    \pgfmathtruncatemacro\nz{\n - 1}
    \pgfmathtruncatemacro\nn{3 * \n + 1}
    \pgfmathtruncatemacro\mm{3 * \m + 1}
    \pgfmathtruncatemacro\nt{3 * \n}
    \pgfmathtruncatemacro\mt{3 * \m}
    \pgfmathtruncatemacro\np{3 * \n - 2}
    \pgfmathtruncatemacro\mp{3 * \m - 2}
    \def\h{1}
    \def\v{1}
    \def\p{3}

    \foreach \k in {1,2,3}{   
    \foreach \j in {1,...,\mm}{
      \foreach \i in {1,...,\nn}{
        \node[vertex] (b\j-\i-\k) at (\j * \h, \k * \p, -\i * \v) {} ;
      }
    }
    }
    
    \foreach \i [count = \ip from 2] in {1,...,4}{
    \foreach \j/\k/\jp/\kp in {1/1/1/2,1/1/1/3,1/2/1/3}{
      \draw[thick] (b\j-\i-\k) -- (b\jp-\ip-\kp) ;
       }
    }

    \foreach \k in {1,2,3}{
    \foreach \j in {1,...,\mm}{
      \foreach \i [count = \ip from 3] in {2,...,\nt}{
        \draw[very thick,\red] (b\j-\i-\k) -- (b\j-\ip-\k) ;
      }
    }
    \foreach \j in {1,...,\mt}{
      \pgfmathtruncatemacro\jp{\j+1}
      \draw[very thick,\red] (b\j-1-\k) -- (b\jp-1-\k) ;
    }
    \foreach \j in {1,3,...,\mt}{
      \pgfmathtruncatemacro\jp{\j+1}
      \draw[very thick,\red] (b\j-\nn-\k) -- (b\jp-\nn-\k) ;
    }
    \foreach \j in {3,5,...,\mt}{
      \pgfmathtruncatemacro\jm{\j-1}
      \draw[very thick,\red] (b\j-2-\k) -- (b\jm-2-\k) ;
    }
    \draw[very thick,\red] (b1-1-\k) -- (b1-2-\k) ;
    \draw[very thick,\red] (b\mm-1-\k) -- (b\mm-2-\k) ;

    \foreach \i in {4}{
     \foreach \j in {1,3,...,\mt}{
      \pgfmathtruncatemacro\jp{\j+1}
      \draw[very thick,red] (b\j-\i-\k) -- (b\jp-\i-\k) ;
     }
    }
    \foreach \i in {3,6,...,\nt}{
     \foreach \j in {3,5,...,\mm}{
      \pgfmathtruncatemacro\jm{\j-1}
      \draw[very thick,red] (b\j-\i-\k) -- (b\jm-\i-\k) ;
     }
    }
    \foreach \j in {4,7,...,\mp}{ 
    \draw[very thick,red] (b\j-1-\k) -- (b\j-2-\k) ;
    }

    \foreach \j in {2,...,\mt}{
    \foreach \i in {5,\nn}{
    \node[filledb] at (\j * \h, \k * \p, -\i * \v) {} ;
      }
    }
    \foreach \j in {1,\mm}{
    \foreach \i in {2,5,...,\nn,\nn,1}{
      \node[filledb] at (\j * \h, \k * \p, -\i * \v) {} ;
    }
    \foreach \i in {3,6,...,\nn}{
    \node[filledb] at (\j * \h, \k * \p, -\i * \v) {} ;
      }
    }
    \foreach \j in {2,5,...,\mm, 3,6,...,\mm}{
      \node[filledb] at (\j * \h, \k * \p, -\v) {} ;
      \node[filledb] at (\j * \h, \k * \p, -2 * \v) {} ;
    }
    \foreach \j in {2,...,\mt}{
      \node[filledp] at (\j * \h, \k * \p, -\nt * \v) {} ;
    }
    \foreach \j in {1,\mm}{
    \foreach \i in {4}{
      \node[filledp] at (\j * \h, \k * \p, -\i * \v) {} ;
    }
    }
    \foreach \j in {4,7,...,\mp}{
      \node[filledp] at (\j * \h, \k * \p, -\v) {} ;
    }
    \foreach \j in {4,7,...,\mp}{
      \node[filledr] at (\j * \h, \k * \p, -2 * \v) {} ;
    }

    }

  \end{tikzpicture}
  }
  \caption{The different layers (instances) linked by the cycle of half-graphs.
  Only four half-graphs are drawn for the sake of legibility.}
  \label{fig:layers}
\end{figure}
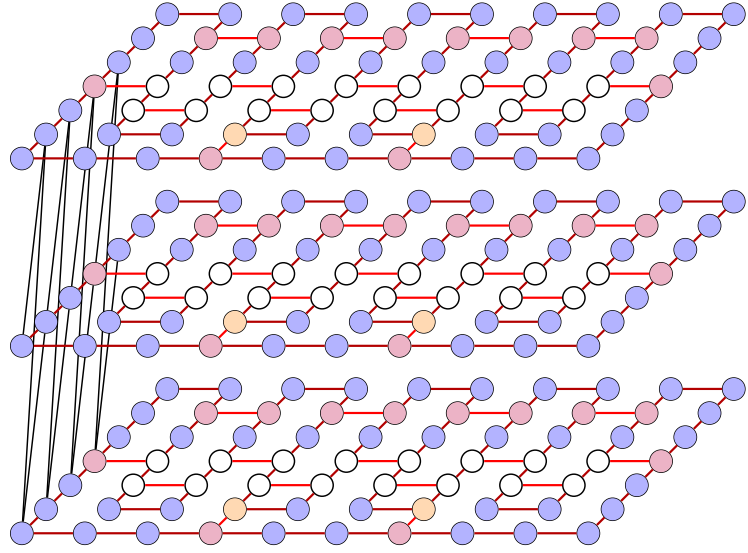

\begin{lemma}\label{lem:contraction-classes}
	The composed graph admits a partial 4-sequence to $H / \left( \bigcup_{i \in [t+1]} \cB_{i} \right)$.
\end{lemma}

\begin{proof}
	\cref{thm:tailored-ds} yields that each input graph $I_i, i \in [t]$ admits a partial 4-sequence to $I_i / \cB_i$ yielding a subgraph of the snaking grid, this is also the case for dummy instance $I_{t+1}$ as an independent set.
	Our contraction sequence will consist exactly of those contractions, so let us see how they behave over $H$. Notice that for any $I_i$, the only contractions made above involve vertices of a single class $B_{i,j}$, which is a module with respect to $H - I_i$. Therefore, in $H$, the partial contraction sequence of each $I_i$ does not create red edges towards $H-I_i$. In turn, we can contract each $I_i$ to $I_i / \cB_i$ through a partial 4-sequence over $H$, and the only resulting red edges will belong to a single instance.
\end{proof}

Now, for the purpose of bounding the twin-width of the composed graph, it is useful to note that \cref{obs:monotone} allows us to add red edges.
At this point, the only vertices of large degree stem from the strict half-graphs.
We keep those edges black, and we now turn each instance into a red augmented snaking grid as depicted in \cref{fig:contraction-order} as follows.
Since each instance is a spanning subgraph of the $p \times q$ snaking grid, we can first assume that it is a (fully) red snaking grid.
Then, the red augmented snaking grid is built by further adding red cycle $(B_{i,1},...,B_{i,N})$.
By our choice of ordering in the composition, this cycle is the same on every instance with respect to their mapping on the $p \times q$ snaking grid.

\begin{lemma}\label{lem:twwbound}
	The composed graph $H$ has twin-width at most 4.
\end{lemma}

\begin{proof}
	By \cref{lem:contraction-classes} and the above remarks, we can now describe the contraction of $t+1$ red augmented snaking grids $(I_i)_{i \in [t+1]}$, abusing notation for the now quotiented instances, with the black edges of our composition. We will exhibit a partial 4-sequence eventually contracting every column, now consisting of $t+1$ homologous vertices, that is, all vertices at the same position on their respective snaking grid into a single one.
	The proof will proceed by induction on the number of augmented snaking grids, our hypothesis at step $t$ being that there exists a partial 4-sequence from $t$ augmented snaking grids to a single one, accounting for the black edges added in the composition. This being true for $t=1$, assume the result holds for some $t$ and let us consider case $t+1$.
	
	We will deal with the two bottommost augmented snaking grids $I_1,I_2$ in the half-graphs, contracting pairs of homologous vertices, corresponding to the quotiented $(B_{1,j},B_{2,j})$ thanks to the ordering chosen in the composition.
	Before proceeding, it is useful to make the following observations:
	\begin{itemize}
		\item For any grid $I_i$ with $i > 2$, the vertices of any homologous pair $(B_{1,j},B_{2,j})$ are each adjacent exactly to $B_{i,j+1}$ in $I_i$. Their contraction yields no red edges, allowing us to omit other instances and only consider the two to be contracted.
		\item The edges involving both snaking grids consist in the matching $(B_{1,j},B_{2,j+1})_j$, see \cref{fig:layers}, which can be assumed red while keeping red degree at most four. In our contraction sequence, this enables us to only bound the total degree at any point.
		Now, when contracting a pair of vertices, any vertex in the above matching adjacent to one of the vertices of our pair is also adjacent to the other because of the additional hamiltonian cycle. Therefore, the neighbors of the contracted pair will be exactly the neighbors of each vertex in its respective instance.
		\item Contracting a pair coarsens the partition of our vertices, so the degree of a vertex not in the pair cannot increase, thus it is enough to bound the degree of the contracted pair. 
	\end{itemize}
	
	Now, consider any pair $(a,b)$ to be contracted, such that both vertices have $C$ already contracted, common neighbors and each $P$ non-contracted neighbors in their respective instance. According to the second observation above, the resulting degree will be $C + 2P$. Since any vertex is of degree at most three in its respective grid $C+P \leq 3$, which would already allow us to contract the pairs in any order with degree at most six.  
	Therefore, to bound the degree by four, the order in which we contract the pairs must be carefully chosen with respect to their position on the respective augmented snaking grids. We want to iteratively contract pairs of vertices satisfying $C+2P \leq 4$, at each step propagating this bound to more non-contracted vertices to eventually contract all pairs.
	We proceed by groups of colors as depicted in \cref{fig:contraction-order}.
	We first contract pairs of blue vertices, those of degree two in their respective augmented snaking grid.
	Second, we contract pairs of purple vertices, those of degree three adjacent to two contracted vertices (necessarily blue).
	Third, we proceed with pairs of orange vertices, those still non contracted on the second row of the snaking grids, which are of degree three and adjacent to a blue and a purple vertex.
	Finally, numbered in white, we will contract the remaining degree three vertices along the corresponding snaking path.
	
	We are now ready to describe the contraction sequence. For the blue, orange and purple pairs, the contractions can be done in any order, while we will describe the order for the last ones.
	For any blue pair, depending on the order of contraction we either have $(C,P)=(0,2)$, $(C,P)=(1,1)$ or $(C,P)=(2,0)$. Thus, regardless of the order the contraction yields degree $C+2P$ at most four.
	For any purple pair, we either have $(C,P)=(2,1)$ or $(C,P)=(3,0)$, again regardless of order contracting the pair yields degree at most four.
	The case is the same for orange pairs, and we can again contract them in any order.
	The graph induced by the non-contracted vertices on each instance now consists in an union of paths, with endpoints satisfying $(C,P)=(2,1)$ and interior points such that $(C,P)=(1,2)$. We can iteratively contract the pairs corresponding to the endpoints, satisfying $(C,P)=(2,1)$, decreasing the length of the path by one at each step. This can be done until we are left with single vertices for each path, the corresponding pairs are then adjacent exactly to three already contracted vertices, and their contraction yields degree three.

	This finishes to prove that the contraction of the first two red snaking grids can be done while bounding the red degree by four. Since the only contracted pairs were homologous, this results in a red augmented snaking grid with no red edges towards grids $i > 2$. The remaining edges of the strict half-graph cycle still form one of height $t$, which is exactly the induction case for $t$ and achieves to prove the induction.
	Therefore, there is a partial 4-sequence from our composed graph into a red augmented snaking grid. Then, as the latter is a subgraph of the red complete grid, \cref{lem:subdivided-grid} yields twin-width at most~4.
\end{proof}

\subsection{Correctness}

Having verified that the reduction runs in polynomial time and satisfies condition $(i)$ of \cref{def:or-composition}, for the OR-cross-composition to be sound it remains to show equivalence of the instances.

\begin{lemma}\label{lem:correctness}
	The composed instance is positive for \kds if and only if at least one input instance is positive for the restricted \kds.
\end{lemma}
 
\begin{proof}
	Consider the OR-cross-composition of $\cR$-equivalent instances $(I_i,N)_i, i \in [t]$ of the restricted \kds into \kds instance $(H,N)$ as described in \cref{sec:composition}. We start by showing the easier direction, the existence of a positive input instance implies that the composed instance is also positive, then we prove the converse. The reduction is such that for any set $D \in H$, $D$ dominates $H$ if and only if it dominates some input instance $I_i$.
	
	\paragraph*{One input instance is positive implies that the composed instance is positive}
	
	Assume there exists some $\ell \neq t+1$ such that $(I_{\ell},N)$ is a positive instance for the restricted \kds, let us consider dominating set $D$ for $I_{\ell}$, with $|D|=N$ and satisfying the conditions of \cref{thm:tailored-ds}. We show that $D$ is also a dominating set for $H$.
	
	Since our reduction preserves the edges of each instance, $D$ dominates $V(I_{\ell})$ over $H$ by assumption.
	To show that $V(H) - V(I_{\ell})= \bigcup_{i \neq \ell,j \in [N]} B_{i,j}$ is dominated, let us consider any cell $B_{i,j}$ with $i \neq \ell$.
	If $i < \ell$, by construction $B_{\ell,j-1},B_{i,j}$ form a biclique in $H$, then $D$ intersects $B_{\ell,j-1}$ via item $(\ref{it3})$ of \cref{thm:tailored-ds} so $B_{i,j}$ is dominated.
	Otherwise, $i > \ell$ and $B_{\ell,j+1}, B_{i,j}$ form a biclique in $H$, now $D$ must intersect $B_{\ell,j+1}$ and again $B_{i,j}$ is also dominated. Therefore $D$ is indeed a dominating set of $H$, yielding the backwards implication.
	
	\paragraph*{The composed instance is positive implies that one input instance is positive}
	
	Conversely, assume $(H,N)$ is a positive instance for \kds, considering any solution $D \subseteq V(H)$, we show that $D$ is necessarily also a dominating set of some $I_{\ell}$, yielding a positive instance among the input ones.
	
	We leverage the fact that $D$ dominates each dummy partition class $B_{t+1,j+1}$ to yield that $|D \cap C_j| \geq 1$ for any $j$.
	By construction, each $B_{t+1,j+1}$ is only adjacent to vertices in column $C_{j}$, so we already know $D$ intersects each $C_j \cup B_{t+1,j+1}$ for $j \in [N]$, and since $|D|=N$ each such intersection consists of a single vertex $d_j$.
	Now, the two vertices of $B_{t+1,j+1}$ are adjacent exactly to $C_j$, so the choice of $d_j$ as one of them would prevent our solution from dominating the other. Then $d_j \in C_j = \bigcup_{i=1}^t B_{i,j}$, letting us define the row choice $r_j$ for column $j$ as the unique $r_j \in [t]$ such that $d_j \in B_{r_j,j}$.
	
	We now show the existence of some $i \in [t]$ such that $D \subseteq V(I_i)$. From the necessity to choose one vertex per column shown above, this is equivalent to showing that for some $i \in [N]$ the row choices are constant: $r_j = i$ for $j \in [N]$.
	We first argue that both the even choice sequence $(r_{2j})_j$ and the odd choice sequence $(r_{2j+1})_j$ must be constant.
	Considered on indices modulo $N$, both sequences are periodic, so if one is not constant we consider a pair of indices $(j,j+2)$ such that corresponding row choices $k=r_j,k''=r_{j+2}$ are such that  $k > k''$.
	
	With $d_j \in B_{k,j}, d_{j+2} \in B_{k'',j}$, let us consider the adjacencies between $D$ and column $C_{j+1}$ towards showing an absurdity.
	On the one hand, the edges added in our composition between $\{d_j,d_{j+2}\}$ and $C_{j+1}$ are exactly $\{(d_j,v)~:~v \in \bigcup_{i > k} B_{i,j+1} \}$ and $\{(d_{j+2},v) : v \in \bigcup_{i < k''} B_{i,j+1} \}$. On the other hand, for any $d \in D \backslash \{d_j,d_{j+2}\}$, any edge between $d$ and $C_{j+1}$ necessarily belongs to some instance $I_i, i \in [t]$.
	Combining the two observations above: for $i \in [k'',k]$, any edge $(a,d)$ with $a \in B_{i,j+1}, d \in D$ must be an edge of $I_i$. 
	Since $k > k''$, interval $[k'',k]$ is of size at least two. We get that $B_{k,j+1}$ is dominated by $D \cap V(I_k)$ in $I_k$, while $B_{k'',j+1}$ is dominated by $D \cap I_{k-1}$ in $I_{k-1}$, see \cref{fig:domination-gap}.
	Now, by \cref{thm:tailored-ds} $(\ref{it3})$, for any $i \in [t]$, $B_{i,j+1} \nsubseteq N(V(I_{i}) \backslash B_{i,j+1})$, yielding that both $D \cap B_{k,j+1}$ and $D \cap B_{k'',j+1}$ are non-empty. This is absurd since $D$ must intersect $C_{j+1}$ in exactly one vertex, so both $(r_{2j})_j$ and $(r_{2j+1})_j$ are constant.
	
	Now, let us call $k$ the row choice for even columns and $k'$ the row choice for odd ones. If $k \neq k'$, we can choose without loss of generality three consecutive columns $C_j,C_{j+1},C_{j+2}$ such that $r_j = r_{j+2} = k$ and $r_{j+1} = k'$ with $k' > k$. 
	Then, we proceed similarly to the last paragraph. Since $k' > k$, $B_{k,j+1}$ is not dominated by $d_j,d_{j+2}$ by construction of the half-graph. Therefore, it is necessarily dominated by $D \cap V(I_k)$ in $H$ and in $I_k$. Again, \cref{thm:tailored-ds} $(\ref{it3})$ yields that $B_{k,j+1} \nsubseteq N(V(I_{k}) \backslash B_{k,j+1})$ so $d_{j+1} \in B_{k,j+1}$, contradicting with $d_{j+1} \in B_{k',j+1}$. Therefore $k=k'$ and the sequence of row choices $(r_j)_j$ is constant.
	
	The latter yields $k \in [t]$ such that $D \subseteq V(I_k)$, by assumption $D$ dominates $V(I_k)$ on $H$, while no edges are added between two vertices in $V(I_k)$ by our construction, thus $D$ is a dominating set of $I_k$.
\end{proof}

\begin{reptheorem}{thm:main}
  Unless {\coNP} $\subseteq$ \NP/poly, \kds on graphs of twin-width at most 4 does not admit a polynomial kernel, even if a 4-sequence of the graph is given.
\end{reptheorem}
\begin{proof}
Lemma~\ref{lem:correctness} proves that the OR-cross-compositions described in the beginning of Section~\ref{sec:composition} is correct. The twin-width bound of the composed instance is given in Lemma~\ref{lem:twwbound}. We can thus invoke Theorem~\ref{thm:kernel-lb} in order to obtain the desired result.
\end{proof}

\begin{figure}[h!]
	\centering
		\begin{tikzpicture}[scale=0.85,vertex/.style={draw,circle,inner sep=0.08cm},picked/.style={fill,circle,opacity=0.5,inner sep=0.055cm}]
			\def\t{9}
			\def\k{5}
			\def\hs{2.6}
			\def\vs{1.3}
			\def\z{0.4}
			\pgfmathtruncatemacro\tm{\t-1}
			\pgfmathtruncatemacro\km{\k-1}
			\pgfmathtruncatemacro\kp{\k+1}
			\pgfmathtruncatemacro\mt{int(\t/2)+1}
			\pgfmathtruncatemacro\mk{int(\k/2)+1}
			\pgfmathtruncatemacro\mtm{int(\t/2)}
			\pgfmathtruncatemacro\mkm{int(\k/2)}
			\pgfmathtruncatemacro\mtp{int(\t/2)+2}
			\pgfmathtruncatemacro\mkp{int(\k/2)+2}

			\pgfmathtruncatemacro\ctop{(\mt+3)*(\vs)}
			\pgfmathtruncatemacro\cbot{(\mt-2)*(\vs)}
			\pgfmathtruncatemacro\cright{(\k+1)*\hs}
			\draw[clip] (0, 2.8) rectangle (\cright, 10.6);
			\begin{scope}
			
			\foreach \i in {1,...,\t}{
				\pgfmathparse{int(\mt - \i)}
				\ifnum\pgfmathresult=0
				\begin{scope}[yshift=\i * \vs cm]
					\foreach \j in {1,...,\k}{
						\begin{scope}[xshift=\j * \hs cm]
							\foreach \s/\x/\y in {1/0/0,2/0/1,3/1/0,4/1/1}{
								\node (a\i\j\s) at (\x * \z,\y * \z) {} ;
							}
							\node at (0.25,0.5*\z) (b\i\j) {$\vdots$} ;
						\end{scope}
					}
					\node at (1.5,0.5 * \z) {$\vdots$} ;
				\end{scope}
				\else
				\begin{scope}[yshift=\i * \vs cm]
					\foreach \j in {1,...,\k}{
						\begin{scope}[xshift=\j * \hs cm]
							\foreach \s/\x/\y in {1/0/0,2/0/1,3/1/0,4/1/1}{
								\node[vertex] (a\i\j\s) at (\x * \z,\y * \z) {} ;
							}
							\pgfmathtruncatemacro\midcheck{abs(\j-\mk)==0}
							\pgfmathtruncatemacro\outercheck{abs(\i-\mt) > 1}
							\pgfmathtruncatemacro\midoutercheck{\midcheck*\outercheck}
							\ifnum\midoutercheck=1
								\node[draw, very thick, rounded corners, fill opacity=0.3, fill=npink!80!black, fit=(a\i\j1) (a\i\j4)] (b\i\j) {} ;
							\else
								\node[draw, very thick, rounded corners, fit=(a\i\j1) (a\i\j4)] (b\i\j) {} ;
							\fi
						\end{scope}
					}
					\node[draw, blue, opacity = 0, very thick, rounded corners, fit=(b\i1) (b\i\k)] (t\i) {} ;
					\pgfmathtruncatemacro\ab{int(\i- \mt)}
					\pgfmathtruncatemacro\abp{int(\ab+1)}
					\pgfmathtruncatemacro\abm{int(\ab-1)}
					\pgfmathtruncatemacro\eqkp{int(\i-\mtm)}
					\pgfmathtruncatemacro\eqk{int(\i-\mtp)}
					\ifnum\ab=1
					\node at (1.5,0.5 * \z) {$I_{k}$} ;
					\else
						\ifnum\ab=-1
							\node at (1.5,0.5 * \z) {$I_{k''}$} ;
						\else
						\fi
					\fi
					\ifnum\ab>1
						\node at (1.5,0.5 * \z) {$I_{k+\abm}$} ;
					\else
					\fi
					\ifnum\ab<-1
						\pgfmathtruncatemacro\abpabs{abs(\abp)}
						\node at (1.5,0.5 * \z) {$I_{k''-\abpabs}$} ;
					\else
					\fi
				\end{scope}
				\fi
			}
			\node[above right = -0.1cm and -0.1cm of a\mtp\mk4] (topInst) {};
			\node[below left = -0.1cm and -0.1cm of a\mtm\mk1] (botInst) {};
			\node[draw, ultra thick, densely dashed, green!50!gray, rounded corners, fit=(topInst) (botInst)] (c) {} ;
			
			\foreach \i in {1,...,\t}{
				\pgfmathparse{int(\mt - \i)}
				\ifnum\pgfmathresult=0
				
				\else
				\begin{scope}[yshift=\i * \vs cm, xshift=\kp * \hs cm]
					\foreach \s/\x/\y in {1/0/0,2/0/1,3/1/0,4/1/1}{
						\node[circle,inner sep=0.08cm] (a\i\kp\s) at (\x * \z,\y * \z) {} ;
					}
					\node[very thick, rounded corners, fit=(a\i\kp1) (a\i\kp4)] (b\i\kp) {} ;
				\end{scope}
				\fi
			}
			
			\foreach \j [count=\jp from 2] in {1,...,\k}{
				\foreach \i in {1,...,\tm}{
					\pgfmathparse{int(\mt - \i)}
					\ifnum\pgfmathresult=0
					\else
					\pgfmathtruncatemacro\ip{\i+1}
					\foreach \iq in {\ip,...,\t}{
						\pgfmathparse{int(\mt - \iq)}
						\ifnum\pgfmathresult=0
						\else
							\ifnum\i=\mtp
								\ifnum\j=\mkm
									\draw[thick,npink!80!black] (b\i\j.east) -- (b\iq\jp.west) ;
								\else
									\draw[black,opacity=0.5] (b\i\j.east) -- (b\iq\jp.west) ;
								\fi
							\else
								\ifnum\iq=\mtm
								          \ifnum\jp=\mkp
										\draw[thick,npink!80!black] (b\i\j.east) -- (b\iq\jp.west) ;
									\else
										\draw[black,opacity=0.5] (b\i\j.east) -- (b\iq\jp.west) ;
									\fi
								\else
									\draw[black,opacity=0.5] (b\i\j.east) -- (b\iq\jp.west) ;
								\fi
							\fi
						\fi
					}
					\fi
				}
			}
			\node[picked,npink!80!black] at (\mkm * \hs+\z,\mtp * \vs+\z) {} ;
			\node[picked,npink!80!black] at (\mkp * \hs,\mtm * \vs) {} ;
		\end{scope}
		\end{tikzpicture}
	        \caption{In the middle column $C_{j+1}$, pink classes are those dominated by the choice of $d_j$ and $d_{j+2}$ in the neighboring columns.
                  At least two classes, highlighted in green and corresponding to the interval $[k'',k]$, would need to be dominated by a single vertex within the column, which cannot be done.}
	\label{fig:domination-gap}
\end{figure}
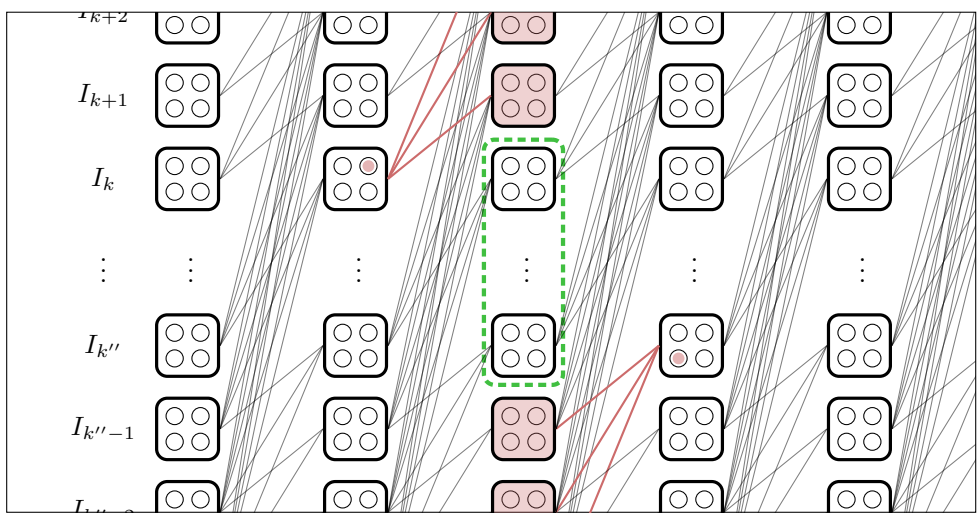

\section{Polynomial kernels}\label{sec:positive}

We present two simple kernelization algorithms for \kconvc\ and \kcapvc\ with kernels on $O(k^2)$ vertices.
We also provide an improved kernel on $O(k^{1.5})$ vertices for \kconvc. 

A folklore  in the context of \textsc{$k$-Vertex Cover} kernelization is that keeping at most $k+1$ twin vertices is a safe reduction rule.
Hence any module outside of a vertex cover $X$ can be trimmed down to $k+1$ vertices.
Similar reduction rules are also known for \kconvc\ and \kcapvc~\cite{CyganGH17}.
Therefore to get a polynomial kernel, one simply needs to polynomially bound the number of modules outside~$X$, or equivalently, the number of distinct neighborhoods in $X$.  

This will be done by proving~\cref{lem:vc-density-1} which, we repeat, is equivalent to saying that the neighborhood hypergraphs of graphs of bounded twin-width have VC density~1.
This feature is shared with classes of bounded expansion.
\cref{lem:vc-density-1} is of independent interest as it opens the door to a common algorithmic treatment for classes of bounded twin-width and of bounded expansion. 

To state the next two theorems, we need to introduce some vocabulary on $0,1$-matrices.
A~\emph{row} (resp.~\emph{column}) \emph{partition} of a matrix $M$ is a partition of its row set (resp. column set).
A~\emph{row} (resp.~\emph{column}) \emph{division} is a row (resp. column) partition where every part is a consecutive set of rows (resp. columns), or \emph{interval}.
A~\emph{division} of a matrix $M$ is a pair $(\mathcal R, \mathcal C)$ where $\mathcal R$ is a row division and $\mathcal C$ is a column division.
A~\emph{cell} or \emph{zone} of a matrix $M$ with division $(\mathcal R, \mathcal C)$ is a submatrix $M[R_i,C_j]$ with $R_i \in \mathcal R$ and $C_j \in \mathcal C$.
A~\emph{$t$-division} is a division $(\mathcal R, \mathcal C)$ with $|\mathcal R|=|\mathcal C|=t$.
A~\emph{$t$-grid minor} of a $0,1$-matrix $M$ is a $t$-division of $M$ where every cell contains at least one 1-entry.
A~matrix is \emph{mixed} if it has at least two distinct rows and at least two distinct columns.
A~\emph{corner} is a $2 \times 2$ contiguous submatrix which is mixed.
A~\emph{$t$-mixed minor} of a matrix $M$ is a $t$-division of $M$ where every cell is mixed.
It can be observed that this is equivalent to the existence of a $t$-division where every cell contains a corner~\cite{twin-width1}.

We first recall the Marcus-Tardos theorem, a celebrated result in combinatorics. 
\begin{theorem}[\cite{MarcusT04}]\label{thm:marcustardos}
For every integer $t$, there is some $c_t$ such that every $n \times m$ $0,1$-matrix~$M$ with at least $c_t\max(n,m)$ entries 1 has a $t$-grid minor.
\end{theorem}
In what follows \emph{$c_t$} will always denote the bound in~\cref{thm:marcustardos}.
The best upper bound for $c_t$ is currently $8/3(t+1)^22^{4t}$~\cite{Cibulka16}, that we take as the definition of $c_t$.

The following result is shown in the first paper of the series dedicated to twin-width.
\begin{theorem}[\cite{twin-width1}]\label{thm:gridtheorem}
    If, for every vertex ordering, the adjacency matrix of a graph $G$ has a $2t+2$-mixed minor, then $G$ has twin-width larger than $t$. 
\end{theorem}

We prove that in a $0,1$-matrix of small twin-width, the number of \emph{distinct} columns is linearly bounded in the total number of rows.
\begin{theorem}\label{thm:linearneigh}
  Let $t$ be a positive integer.
  Let $M$ be a $0,1$-matrix with $s$ rows, and at least $2^{4c_{2t+2}} \cdot s$ distinct columns.
  Then $M$ has a $2t+2$-mixed minor. 
\end{theorem}
\begin{proof}
We denote by $R$ and $C$ the sets of rows and columns, respectively, of $M$.
Without loss of generality, we may assume that all the columns of $M$ are distinct.
We consider a division of $C$ into $s$ parts $\{C_1, \ldots, C_s\}$, each $C_i$ consisting of $2^{4c_{2t+2}}$ consecutive (distinct) columns. 
Note that the submatrix of $M$ consisting of the columns $C_i$ has rank at least $4c_{2t+2}$ in the binary field $\mathbb F_2$, for all $i \in [s]$. 

Therefore there is a row division of $M[R,C_i]$ into at least $2c_{2t+2}$ (row) parts, each zone of which has rank at least~2, thus is mixed, and, by an observation in~\cite{twin-width1}, contains a corner. 
These $2c_{2t+2}$ corners in $M[R,C_i]$ are on pairwise disjoint pairs of consecutive rows.
Let us consider $\mathcal R^1$ the row division grouping each pair of rows with indices $2i-1, 2i$, and $\mathcal R^2$, grouping each pair of rows $2i, 2i+1$, for $i \in [\lceil s/2 \rceil]$.
Observe that one of the two divisions $(\mathcal R^1,\{C_i\})$, $(\mathcal R^2,\{C_i\})$ of $M[R,C_i]$ contains at least $c_{2t+2}$ zones with a corner, hence mixed.

Without loss of generality, we may assume that at least $\lceil s/2 \rceil$ column parts among $C_1, \ldots, C_s$ have at least $c_{2t+2}$ mixed zones when divided by, say, $\mathcal R^1$.
Consider the column division $\{C'_1, \ldots, C'_{s'}\}$ with $s' \geq \lceil s/2 \rceil$, coarsening of $\{C_1, \ldots ,C_s\}$ such that each part $C'_j$ contains exactly one column part $C_i$ with the property of the previous sentence. 

Let $M'$ be a $|\mathcal R^1| \times s'$ $0,1$-matrix with a 1-entry at positions where the cell of $(\mathcal R^1,\{C'_1, \ldots,$ $C'_{s'}\})$ is mixed, and a 0-entry otherwise.
Note that $M'$ has at most $\lfloor s/2 \rfloor$ rows and $s' \geq \lceil s/2 \rceil$ columns. 
Moreover each column of $M'$ contains by design at least $c_{2t+2}$ 1-entries. 
By~\cref{thm:marcustardos}, $M'$ admits a $(2t+2)$-grid minor.
Thus $M$ has a $(2t+2)$-mixed minor, which implies that $M$ has twin-width larger than $t$ by~\cref{thm:gridtheorem}. 
\end{proof}


We conclude the following.
\begin{replemma}{lem:vc-density-1}
  For every graph $G$ of twin-width~$t$ and $X \subseteq V(G)$, the number of distinct neighborhoods in $X$, $|\{N(v) \cap X~:~v \in V(G)\}|$, is at most $2^{4c_{2t+2}}|X|$.
\end{replemma}
\begin{proof}
  We assume for the sake of contradiction that $|\{N(v) \cap X~:~v \in V(G)\}|>2^{4c_{2t+2}}|X|$.
  For every vertex ordering of $G$, its adjacency matrix $M$ along this order contains an $|X| \times 2^{4c_{2t+2}}$ submatrix without two equal columns; namely the submatrix of the adjacencies between $X$ and $2^{4c_{2t+2}}$ vertices with a pairwise distinct neighborhood in $X$.
  By~\cref{thm:linearneigh}, it implies that $M$ has a $2t+2$-mixed minor.
  By~\cref{thm:gridtheorem}, this in turn implies that $G$ has twin-width more than $t$.
\end{proof}
  

\subsection{Quadratic vertex kernels}

For completeness, we state and prove the folklore reduction rule for \textsc{$k$-Vertex Cover} variants.

\begin{reduction}[Reduction Rule for \kconvc]\label{rule:convc}
Let $X$ be a vertex cover of $G$. 
If there is a set $S \subseteq V(G) \setminus X$ with the same neighborhood in $X$ and $\abs{S}>k$, delete a vertex of $S$. 
\end{reduction}

\begin{lemma}\label{lem:safeconvc}
Let $G$ be a graph and $G'$ be a graph obtained by applying~\cref{rule:convc}. Then $(G,k)$ is a \textsc{yes}-instance to \kconvc\ if and only if $(G',k)$ is. 
\end{lemma}
\begin{proof}
  Let $s \in S$ be the vertex such that $G':=G-s$.
  Suppose that $T$ is a connected vertex cover of $G$ of size at most $k$.
  If $s \notin T$, then  $T$ is a connected vertex cover of $G'$.
  If $s \in T$, there is at least one vertex $s'$ in $S \setminus T$.
  Now $T':=(T \setminus \{s\})\cup \{s'\}$ is a connected vertex cover of $G'$ of size $|T|$.
  
  Conversely let $T$ be a connected vertex cover of $G'$ of size at most $k$.
  We claim that $T$ is also a connected vertex cover in $G$.
  As $G[T]$ is connected and $|T|\leq k$, we shall just check that all the edges of $G$ are covered by $T$.
  This is because $S\setminus s$ cannot be totally included in $T$; if so, $T$ is not connected because $S\setminus s$ consists of (at least) $k$ pairwise independent vertices. 
  Therefore, the common neighborhood $N(S)$ should be included in $T$.
\end{proof}

\begin{proposition}\label{prop:connkernel}
  \kconvc\ admits a kernel on $O_t(k^2)$ vertices when the input graphs have twin-width at most $t$. 
\end{proposition}
\begin{proof}
Let $(G,k)$ be an instance of \kconvc\ and let $X$ be a vertex cover obtained using a 2-approximation algorithm for \textsc{Min Vertex Cover} (given by any maximal matching). 
If $\abs{X}\geq 2k+1$, an optimal vertex cover and thus an optimal connected vertex cover has size at least $k+1$. 
Then one can correctly output a trivial \textsc{no}-instance. 
Henceforth we assume that $\abs{X}\leq 2k$.
Apply~\cref{rule:convc} exhaustively, and let $(G',k)$ be the resulting instance. 
By~\cref{lem:vc-density-1} and by the construction of $G'$, $V(G') \setminus X$ can be partitioned into at most $2^{4c_{2t+2}} \cdot \abs{X}$ modules and each module consists of at most $k+1$ vertices.
Therefore, $V(G')$ has at most $c_t \cdot (k+1)+2k$ vertices. 
That $(G',k)$ is an equivalent to $(G,k)$ is implied by~\cref{lem:safeconvc}. 
\end{proof}

\begin{reduction}[Reduction Rule for \kcapvc]\label{rule:capvc}
Let $(G,k,c:V(G) \rightarrow \mathbb{N})$ be an instance to \kcapvc\ and $X$ be a vertex cover of~$G$. 
If there is a set $S \subseteq V(G) \setminus X$ with the same neighborhood in $X$ and $\abs{S}>k+1$, delete a vertex of $s\in S$ with the minimum capacity and decrease the capacity of each neighbor of $s$ by one. 
\end{reduction}

\begin{lemma}\label{lem:safecapvc}
Let $(G,k,c)$ be an instance to \kcapvc\ and $(G',k,c')$ be an instance obtained by applying~\cref{rule:capvc}. 
Then $(G,k,c)$ is a \textsc{yes}-instance if and only if $(G',k,c')$ is. 
\end{lemma}
\begin{proof}
  Again let $s$ be such that $G' := G-s$.
  We first check the forward direction.
  A $c$-capacitated vertex cover of $G$ is not necessarily a $c'$-capacitated vertex cover of $G'$ as we decreased the capacities of neighbors of $s$.
  However we claim that if $G$ admits a $c$-capacitated vertex cover $T$ of size at most $k$, then it also admits a $c$-capacitated vertex cover of $G$ of size at most $k$ not containing $s$.
  Indeed, if $s \in T$, then there exists a vertex $s' \in S \setminus T$ and not that all neighbors of $s'$ (thus $s$) are in $T$.
  Therefore, $(T\setminus \{s\})\cup \{s'\}$ is a vertex cover of $G$ of size at most $k$.
  It is further $c$-capacitated because $c(s') \geq c(s)$.
  Now we may assume that $T$ does not contains $s$.
  Set $T$ is a $c'$-capacitated vertex cover of $G'$ because the decrease of capacity of each vertex $v \in N(s)$ by one unit is canceled out by the absence of edge $sv$ in $G'$. 

  Conversely, if there is a $c'$-capacitated vertex cover $T$ of $G'$ of size at most $k$, note that all neighbors of $S \setminus \{s\}$ must be in $T$.  
  Therefore, $T$ is a $c$-capacitated vertex cover of $G$ where each edge incident with $s$ is covered by the residual capacity of the other endpoint.
\end{proof}

Therefore, \kcapvc\ admits a quadratic vertex kernel, which follows the proof of~\cref{prop:connkernel} verbatim.
\begin{proposition}\label{prop:cpvc}
  \kcapvc\ admits a kernel on $O_t(k^2)$ vertices when the input graphs have twin-width at most $t$.
\end{proposition}

\cref{thm:cnvc-cpvc} is a direct consequence of~\cref{prop:connkernel,prop:cpvc}.

\subsection{Improved kernel for \kconvc}

We present here a kernelization algorithm for \kconvc\ on bounded twin-width graphs which leads to an instance on $O(k^{1.5})$ vertices, and a simple linear kernel when the input class is further restricted to be sparse, that is, $K_{s,s}$-free. 

Let $X$ be a vertex cover of $G$, and let $X^b$ (resp.  $X^s$) be the subsets of $X$ containing all vertices of $X$ with at least $k+1$, respectively at most $k$, neighbors in $V(G) \setminus X$. 
Let $Y_1, \ldots, Y_q$ be the partition of $V(G) \setminus X$ into maximal modules. 
For each $i \in [q]$, let $X_i$ be the neighbors of $Y_i$ in $X^s$.

\begin{reduction}\label{rule:convc2}
If there is $i\in [q]$ with $X_i\neq \emptyset$ and $\abs{Y_i}\geq \abs{X_i}+2$, then delete a vertex of $Y_i$. 
\end{reduction}

\begin{lemma}\label{lem:safeconvc2}
  Let $G$ be a graph and $G'$ be a graph obtained by applying~\cref{rule:convc2}.
  Then $(G,k)$ is a \textsc{yes}-instance to \kconvc\ if and only if $(G',k)$ is. 
\end{lemma}
\begin{proof}
  Let $Y_i \subseteq V(G) \setminus X$ be a module within which~\cref{rule:convc2} is applied, and $y \in Y_i$ be the deleted vertex, that is, such that $G':=G-y$.
  Set $Y'_i := Y_i \setminus \{y\}$.
  We first observe that the vertices of $X^b$ are mandatory for any feasible solution to $(G',k)$.
  Let $T$ be an arbitrary connected vertex cover of $G'$ of size at most $k$.
 
\begin{claim}\label{claim:mandatory}
$X^b \subseteq T$.
\end{claim}
\begin{proofofclaim}
Suppose some vertex $w\in X^b$ is not contained in $T$ and thus all its neighbors belong to $T$. 
Note that $w$ has at least $k$ neighbors outside $X$ in $G'$, even when the deleted vertex $y$ was a neighbor of $w$ in $G$. 
We conclude that $T=N_{G'}(w)\setminus X$, 
which contradicts that $G'[T]$ is connected.
\end{proofofclaim}

Because there are all possible edges between $Y'_i$ and $X_i$, $T$ must contain at least one of $Y'_i$ and $X_i$ entirely. 
The next claim says that $T$ can be assumed to fully contain $X_i$.

\begin{claim}\label{claim:replaceY}
Let $y'$ be a vertex of $Y'_i$. 
If $Y'_i\subseteq T$, then $T':=(T\setminus Y'_i)\cup (X_i\cup \{y'\})$  is a connected vertex cover of $G'$ of size at most $k$.
\end{claim}
\begin{proofofclaim}
Notice that $X^b\cup X_i$ is  the neighborhood of $Y'_i$ in $G'$. 
By~\cref{claim:mandatory}, $T'$ is clearly a vertex cover of $G'$  and $y'$ provides any connection between a pair of vertices in $T'$ that $Y'_i$ used to provide. 
Finally, that $Y'_i$ is obtained from $Y_i$ after~\cref{rule:convc2} means that $\abs{Y'_i}\geq \abs{X_i}+1=\abs{X_i\cup \{y'\}}$, and 
thus $\abs{T'}\leq \abs{T}$. 
\end{proofofclaim}

Due to~\cref{claim:replaceY}, we may assume that $T$ contains $X_i$. Now for any edge $xy$ incident with the deleted vertex $y$, 
$x$ is either in $X^b$ or $X_i$, and thus the edge $xy$ is covered by $T$ by~\cref{claim:mandatory} and the assumption $X_i\subseteq T$. 
It  follows that $T$ is a feasible solution to $(G,k)$.

Conversely, let $T$ be a connected vertex cover of $(G,k)$ of size at most $k$. Following the same arguments as  the above  claims, one can easily check 
that $T$ must contain both $X^b$ and $X_i$, and contains at most one vertex of $Y_i$. As we can modify $T$ so that it does not contain $y$, 
$T$ is a feasible solution to $(G',k)$. 
\end{proof}

\begin{proposition}\label{prop:connkernel2}
  \kconvc\ admits a kernel on $O_t(k^{1.5})$ vertices when the input graphs have twin-width at most $t$. 
\end{proposition}
\begin{proof}
  Let $(G,k)$ be the input instance of \kconvc.
  We can safely remove any isolated vertex, and assume that every connected component of $G$ contains at least one edge.
If $G$ contains more than one connected component, then clearly $(G,k)$ is a \textsc{no}-instance and we output the 4-vertex graph with two isolated edges. 
Therefore, we can assume that $G$ is connected.
With a 2-approximation algorithm for \textsc{Vertex Cover}, one can find a vertex cover $X$ of $G$ and assume that $\abs{X}\leq 2k$.
Indeed if this is not the case, we can correctly output a trivial \textsc{no}-instance because $G$ does not admit a connected vertex cover of size at most $k$. 

Note that~\cref{rule:convc2} does not disconnect the given graph as we remove a vertex only when it has a twin. 
Let $(G',k)$ be an instance obtained by exhaustively applying~\cref{rule:convc2} with the vertex cover $X$ at hand. 
We classify $X$ into $X^b$ and $X^s$ as before, and  $Y_1, \ldots, Y_q$ denote the partition of $Y:= V(G') \setminus X$ into maximal modules. 
For each $i \in [q]$, $X_i$ is the neighbors of $Y_i$ in $X^s$.
By~\cref{lem:vc-density-1}, we have $q \leq 2 \cdot 2^{4c_{2t+2}} k$.
Because the edge set between $X^s$ and $Y$ is decomposed into the edge sets of complete bipartite graphs on $(Y_i,X_i)$ over $i\in [q]$, the number of edges between $X^s$ and $Y$ is at least 
$$\sum_{i=1}^q \abs{Y_i}\cdot \abs{X_i}\geq  \sum_{i=1}^q (\abs{Y_i}-1)^2 \geq \frac{1}{q}\cdot \left(\sum_{i=1}^q (\abs{Y_i}-1) \right)^2 \geq \frac{1}{q}\cdot (\abs{Y}-q)^2.$$
Suppose that $\abs{Y}-q> 2\cdot 2^{2c_{2t+2}}\cdot k^{1.5}$. 
Now, 
$$\frac{1}{q}\cdot (\abs{Y}-q)^2 > \frac{4\cdot 2^{4c_{2t+2}} \cdot k^3}{2\cdot 2^{4c_{2t+2}}\cdot k}=2k^2,$$
and hence there are more than $2k^2$ edges between $X^s$ and $Y$.
With $\abs{X^s}\leq 2k$, this implies that  
there exists a vertex in $X^s$ which has more than $k$ neighbors in $Y$, contradicting the definition of $X^s$. 
To conclude, $G'$ has at most 
$$\abs{X}+\abs{Y}\leq 2k + 2\cdot 2^{2c_{2t+2}}\cdot k^{1.5} + q \leq   2k + 2\cdot 2^{2c_{2t+2}}\cdot k^{1.5} + 2\cdot 2^{4c_{2t+2}} \cdot k = O_t(k^{1.5})$$ vertices as claimed. 
\end{proof}

Note that the proof of~\cref{prop:connkernel2} only uses the fact that the input graphs have VC density at most~1, so we in fact established~\cref{thm:cnvc-improved}.

\section{Graphs of twin-width~1}\label{sec:poly}

Cycles on at least five vertices and their complements have twin-width~2.
Thus graphs of twin-width 1 do not have induced cycles of length at least 5, nor their complements.
In particular, they are perfect.
Hence computing the independence, clique, and coloring numbers can be done in polynomial-time~\cite{Grotschel84}.

Graphs of twin-width~0 are cographs, and can be recognized in linear time~\cite{Habib05}. 
It turns out that twin-width~1 graphs are also efficiently recognizable.
For that, we need the following technical lemma.
\begin{lemma}\label{lem:one-red-edge}
  Let $G$ be a prime graph of twin-width 1, and let $G=G_n, \ldots, G_2, G_1$ be a 1-sequence.
  Then the trigraphs $G_i$ have exactly one red edge, except for $G_n$ and $G_1$ which have no red edge.
\end{lemma}
\begin{proof}
  The idea of the proof is to start with the 3-vertex trigraph $G_3$ and to rewind the contraction sequence back to $G_n=G$.
  Let $a, b, c$ be the three vertices of $G_3$.
  Since $G_3$ is a 1-trigraph, there is at most one red edge in $G_3$.
  Since $G$ is a prime graph, there has to be at least one red edge in $G_3$, say between $b$ and $c$.
  Furthermore, $a$ has to be a vertex in $G$, otherwise $a(G)$ is a module of size at least two in $G$.
  Moreover, $a$ is adjacent to, say, $b$, but non-adjacent to the other vertex, $c$ since otherwise $b(G)\cup c(G)$ is a module.
  
  We show by induction the following property $(\mathsf P)$ for $i \in [2,n-1]$: 
  \begin{center}
  {\sf Property P}: in  $G_i$, there is exactly one red edge, say, $uv$ ($u$ and $v$ can depend on $i$) and all the vertices that are not $u$ nor $v$ are original vertices of $G$.
  \end{center}
  We already observed that  {\sf Property P} holds  for $G_3$. It  also holds for $G_2$ because $G$ contains no non-trivial module. 
  Assuming that $G_i$ satisfies {\sf Property P} for some $i \in [3,n-2]$, 
  let us prove that $G_{i+1}$ also satisfies {\sf Property P}.
  Let $uv$ be the red edge in $G_i$.
  By assumption, all the other vertices are vertices in $G$.
  Hence, the graph $G_{i+1}$ is obtained from $G_i$ by splitting $u$ or $v$, say $u$, into two vertices $u_1$ and $u_2$.
  Now, there is at most one red edge in $G_{i+1}[\{u_1,u_2,v\}]$ since otherwise a vertex shall be incident with two  red edges. 
  Recall that at least one vertex of $G_{i+1}$, thus one of $\{u_1,u_2,v\}$ is not an original vertex of $G$. Such a vertex must  
  be incident with a red edge because it is not a module. Therefore, $G_{i+1}[\{u_1,u_2,v\}]$ has a unique red edge. 
  Without loss of generality, assume that the red edge is $u_2v$.
  Again, $u_1$ should be an original vertex of $G$, i.e., $|u_1(G)|=1$, because it has no incident red edge and therefore $u_1(G)$ forms a module.
  So {\sf Property P} holds for $G_{i+1}$.
\end{proof}

We can efficiently detect graphs of twin-width~1 using the previous lemma and an inductive scheme.
\begin{theorem}\label{thm:tww-one}
  Twin-width 1 graphs can be recognized in polynomial time.
  Moreover, a 1-sequence where each trigraph has at most one red edge can be constructed in polynomial time.
\end{theorem}
\begin{proof}
  Let $G$ be a given graph. 
  We show by induction on $n$, the number of vertices of $G$, how to decide whether or not it has twin-width 1.
  If $G$ has one vertex, it has twin-width 1 (even 0).
  If $G$ is not a prime graph, then each maximal module of $G$ and the quotient graph of $G$ by its (pairwise disjoint) maximal modules all have at most $n-1$ vertices.
  So \cref{lem:modular} and the inductive step allow us to conclude.
  We can therefore assume that $G$ is a prime graph.
  
  As a technicality, it is more convenient to assume that $G$ is a trigraph with exactly one red edge.
  We can in quadratic time guess the first contraction (which creates exactly one red edge since the contractions of twins are safe by~\cref{obs:induced-subgraph}).
  Recall that $G$ has exactly one red edge, say $uv$.
  Let us establish the next contraction.
  By \cref{lem:one-red-edge}, we can limit ourselves to contractions which keep the number of red edges to at most one.
  Therefore, the next contraction cannot be between two vertices in $V(G) \setminus \{u,v\}$.
  Indeed, that would either add at least one new red edge to the trigraph or the two contracted vertices would be twins, and as such, in a module of size at least two.

  Any contraction between $w \in V(G) \setminus \{u,v\}$ and $u$ (or $v$) that creates a new red edge can be disregarded again by \cref{lem:one-red-edge}.
  The only way that such a contraction does not increase the number of red edges is if the contraction actually results in the induced subgraph $G - \{w\}$.
  By~\cref{obs:induced-subgraph}, this contraction can be safely performed.
  When no such safe contraction is available, we know that only the current red edge can possibly be contracted.
  Hence, we can safely make this contraction, and we obtain a new graph with $n-1$ vertices.
  Observe that in this process, we never branch on two possible contractions.
  
  The second statement immediately follows by concatenating the 1-sequence of the modules and of the quotient graph, each of twin-width at most 1, as in~\cref{lem:modular}.
\end{proof}

We remark that polynomial time algorithms are not hard to obtain for, say, \kconvc, \kds\ or \textsc{Connected $k$-DS}, when a 1-sequence as in~\cref{thm:tww-one} 
is given. Consider \kconvc\ for instance. For each red component, we "trace" a partial solution $X$, i.e., a vertex set covering all original edges whose endpoints belong to the said red component, 
by recording for each vertex $u$ of the component whether 
$u(G)$ is fully contained in $X$, intersecting (but not fully contained in) $X$, or not intersecting at all. Among all partial solutions leaving the same trace over the red component, we remember the minimum size of a solution. 
One can readily check that the optimal value of a partial solution per trace can be updated in constant time for each newly created red component. 

Readers may 
further notice that this dynamic programming works not only when there is at most one red edge, but more generally when a red component size is bounded by a constant.
In a manuscript under preparation, the authors show that MSOL$_1$-expressible problems 
can be solved in $f(d)\cdot n^{O(1)}$-time if the input graph admits a contraction sequence in which each red component has size at most $d$. 

\paragraph*{Acknowledgments}
We thank Noga Alon and Bart M.~P.~Jansen for independently asking whether \kds admits a polynomial kernel on classes of bounded twin-width, an interesting question that led to our main result.

\end{document}